\documentclass[twoside]{ima}
\usepackage{amssymb,epsf}
\usepackage{color,graphicx,arydshln}
\usepackage[hmargin=1in,vmargin=1in]{geometry}
\usepackage{ulem}
\renewcommand{\a}{\displaystyle}

\newcommand{\Fp}{{\lower.10em\hbox{${\mathbb F}$}}\kern-.1em\hbox{}_p}

\begin{document}
\title{Emergence of Anti-Cancer Drug Resistance: Exploring the Importance of the Microenvironmental Niche via a Spatial Model\thanks{This work was supported by the Institute for Mathematics and Its Applications (IMA) as a part of the Research Collaboration Workshop for Women in Applied Mathematics, Dynamical Systems and Applications to Biology and Medicine--WhAM!}}

\author{Jana L. Gevertz\thanks{Department of Mathematics and Statistics, The College of New Jersey, Ewing, NJ, USA.}
\and Zahra Aminzare\thanks{Mathematics Department, Rutgers University, Piscataway, NJ, USA.}
\and Kerri-Ann Norton\thanks{Department of Biomedical Engineering, Johns Hopkins University, Baltimore, MD, USA. This author was supported by T32 CA13084005 and an American Cancer Society postdoctoral fellowship.}
\and Judith P\'erez-Vel\'azquez\thanks{Helmholtz Zentrum M\"unchen, German Research Center for Environmental Health, Institute of Computational Biology, Munich, Germany. This author would like to thank the German Research Foundation (Deutsche Forschungsgemeinschaft, DFG) for providing a travel grant (CE 243/1-1) to facilitate the international cooperation.}
\and Alexandria Volkening\thanks{Division of Applied Mathematics, Brown University, Providence, RI, USA. This author was supported by NSF Graduate Research Fellowship (GRFP) under Grant No. DGE0228243.}
\and Katarzyna A. Rejniak\thanks{Integrated Mathematical Oncology Department, H. Lee Moffitt Cancer Center and Research Institute, Tampa, FL, USA \& Department of Oncologic Sciences, College of Medicine, University of South Florida, Tampa, Florida, USA. The work of this author was supported in part by the NIH grants U54-CA113007 and U54-CA-143970, and the Institutional Research Grant number 93-032-16 from the American Cancer Society.}
}

\maketitle

\begin{abstract} 
Practically, all chemotherapeutic agents lead to drug resistance. Clinically, it is a challenge to determine whether resistance arises prior to, or as a result of, cancer therapy. Further, a number of different intracellular and microenvironmental factors have been correlated with the emergence of drug resistance. With the goal of better understanding drug resistance and its connection with the tumor microenvironment, we have developed a hybrid discrete-continuous mathematical model. In this model, cancer cells described through a particle-spring approach respond to dynamically changing oxygen and DNA damaging drug concentrations described through partial differential equations.  We thoroughly explored the behavior of our self-calibrated model under the following common conditions: a fixed layout of the vasculature, an identical initial configuration of cancer cells, the same mechanism of drug action, and one mechanism of cellular response to the drug. We considered one set of simulations in which drug resistance existed prior to the start of treatment, and another set in which drug resistance is acquired in response to treatment. This allows us to compare how both kinds of resistance influence the spatial and temporal dynamics of the developing tumor, and its clonal diversity. We show that both pre-existing and acquired resistance can give rise to three biologically distinct parameter regimes: successful tumor eradication, reduced effectiveness of drug during the course of treatment (resistance), and complete treatment failure.  When a drug resistant tumor population forms from cells that acquire resistance, we find that the spatial component of our model (the microenvironment) has a significant impact on the transient and long-term tumor behavior.  On the other hand, when a resistant tumor population forms from pre-existing resistant cells, the microenvironment only has a minimal transient impact on treatment response. Finally, we present evidence that the microenvironmental niches of low drug/sufficient oxygen and low drug/low oxygen play an important role in tumor cell survival and tumor expansion. This may play role in designing new therapeutic agents or new drug combination schedules.  
\end{abstract}

\begin{keywords} tumor therapy, tumor environment, hybrid model, individual cell-based model. \end{keywords}

{\AMSMOS 92C50, 93C37. \endAMSMOS}

\section{Introduction}
The effectiveness of practically all chemotherapeutic compounds that are used in the clinical treatment of solid tumors reduces during the course of therapy \cite{KimTannock:2005,DeanEtAl:2005,Baguley:2010,ZahreddineBorden:2013}.  This phenomenon, termed drug resistance, can arise due to a number of different intracellular and microenvironmental causes. Drug resistance is not cancer-specific; it is a common cause of treatment failure in HIV infection \cite{BockLengauerEtAl:2012}, and of antibiotic failure in bacterial communities \cite{LambertEtAl:2011}. It is often impossible to determine whether the drug resistant cells existed prior to the start of treatment, or if they arose as a result of anti-cancer therapy. Pre-existing (primary) drug resistance occurs when the tumor contains a subpopulation of drug resistant cells at the initiation of treatment, and these cells become selected for during the course of therapy. Acquired resistance involves the adaptation of a tumor cell subpopulation so that the cells gradually develop drug resistance due to drug exposure and other factors, such as microenvironmental or metabolic conditions \cite{TredanEtAl:2007,WuEtAl:2013}.

Mechanisms of drug resistance are being studied in cell culture \cite{RottenbergEtAl:2007,CorreiaBissell:2012}. Drug resistant cell lines are produced by incubating cells with a particular drug, collecting the surviving cell subpopulation, and repeating this process through several passages until the remaining subpopulation of cells no longer responds to the treatment. While this is an effective way to generate a resistant cell population, this {\it in vitro} process does not reveal whether the surviving cells become less responsive to chemotherapeutic treatment with each cell passage (acquired drug resistance), or if a small population of resistant cells was present from the beginning and simply overgrew the other cells during the course of the experiment (pre-existing drug resistance). 

Many different mechanisms can be adopted by cancer cells to resist treatment \cite{DeanEtAl:2005,Baguley:2010,ZahreddineBorden:2013,Lage:2008,Cheung-OngEtAl:2013,WoodsTurchi:2013}.  These can broadly be divided into intrinsic (intracellular) causes and extrinsic (microenvironmental) causes. As an example of intrinsic resistance, a cancer cell arrested in a quiescent state will not respond to the killing effects of an anti-mitotic drug.  Other intrinsic causes of drug resistance include enhanced DNA repair mechanisms, increased tolerance to DNA damage, high levels of drug transporters that eliminate drug from the cell, over-expression of drug target receptor, or accumulation of cancer stem cells. Among the extrinsic causes of drug resistance are factors that synergistically limit cancer cell exposure to drug.  These include irregular tumor vasculature that causes chaotic drug delivery and interstitial fluid pressure that hinders drug transport.  Metabolic gradients inside tissue, such as regions of hypoxia or acidity, can also influence cell sensitivity to a drug \cite{TredanEtAl:2007,CorreiaBissell:2012,MeadsEtAl:2009,NakasoneEtAl:2012,McMillinEtAl:2013}. While these factors could be pre-existing, exposure to drug has also been found to affect these different factors \cite{KimTannock:2005,TredanEtAl:2007,ProvenzanoHingorani:2013}. 

To narrow down the focus of this paper, we only consider a chemotherapeutic agent that chemically reacts with and damages cell DNA. Among these are drugs routinely used in the clinic, including alkylating agents ({\it cisplatin}, {\it mephalan}), antimetabolites ({\it 5-fluorouracil}, {\it gemcitabine}), anthracyclines ({\it doxorubicin}) or topoisomerase poisons ({\it etoposide}). DNA integrity is essential for a cell to properly function, and when increased levels of DNA damage are detected at cell-cycle checkpoints, the cell can be arrested in its cell-cycle to give time for DNA repair. However, in cancer cells the mechanisms of DNA damage sensing can be loosened, and some cells are capable of ignoring cell-cycle checkpoints. We consider cancer cells with one or more of these features to be \textit{resistant} to DNA damaging drugs. These resistant cells can result in the increased proliferation and replication of the damaged DNA. Cancer cells can also develop resistance through increased DNA repair capabilities and increased DNA damage tolerance \cite{Cheung-OngEtAl:2013,WoodsTurchi:2013,Karran:2001,SawickaEtAl:2004}. Clinically, it is generally not possible to determine which of these mechanisms result in resistance to DNA damaging drugs, though all of these mechanisms have been observed experimentally.

Our goal is to simulate how resistance to DNA damaging drugs (administered continuously by intravenous injection) can emerge in individual cells and in a growing population of tumor cells. We were interested in comparing tumor dynamics in two cases: when a small subpopulation of tumor cells is already resistant to the drug (pre-existing resistance) and when the cells can become resistant upon exposure to the drug (acquired resistance). Theoretical analysis of the temporal components of the model, along with numerical simulations of the full spatial model have been performed.  These analyses revealed multiple biologically-distinct parameter regimes for both pre-existing and acquired resistance.  In either case, we found a parameter regime for which the effectiveness of drug is reduced during the course of treatment (drug resistance). Focusing on the case of pre-existing resistant cells, we quantified how their location in tissue space influences therapeutic response. When drug resistance emerges in response to treatment, we explored which clonal populations survived treatment and how this depended on their location in tissue space. Moreover, we present evidence that two different microenvironmental niches, low drug/low oxygen and low drug/sufficient oxygen, have different transient and long-term impacts on cancer cell survival.

\section{Mathematical Model} \label{sec:MathModel}

We have developed a hybrid discrete-continuous model of a two-dimensional tissue slice in which a tumor grows, interacts with the microenvironment and is treated with a DNA damaging drug. The pivotal role of the microenvironment in drug efficacy and resistance is strongly suggested through the observation that drugs with potent {\it in vitro} activity are often significantly less effective in the clinic. We incorporate the microenvironment by considering a small patch of tissue with imposed positions of non-evolving blood vessels that serve as a source of both nutrients and drug. We track the features and behaviors of individual cells; this includes their clonal evolution and their interactions with other cells and the surrounding microenvironment. For simplicity, we do not include any stromal cells or other extracellular components. The hybrid discrete-continuous model, described in more detail below, combines an agent-based technique (a particle-spring model) to represent the individual tumor cells and continuous partial differential equations to describe the kinetics of both oxygen and drug. The cellular and microenvironmental components of our model are shown graphically in {\sc Fig.}\ref{fig:FigColor}a. A comparison of our model with other hybrid models of tumor resistance to drugs is given in the Discussion in Section \ref{sec:Discussion}. In all equations below we will use the following notation: $\mathbf{x}=(x,y)$ will define locations of continuous variables, such as drug and oxygen concentrations, while $(X,Y)$ will denote locations of discrete objects, that is cells and vessels.

\begin{figure}[ht] 
\begin{center}
\includegraphics[width=0.9\textwidth]{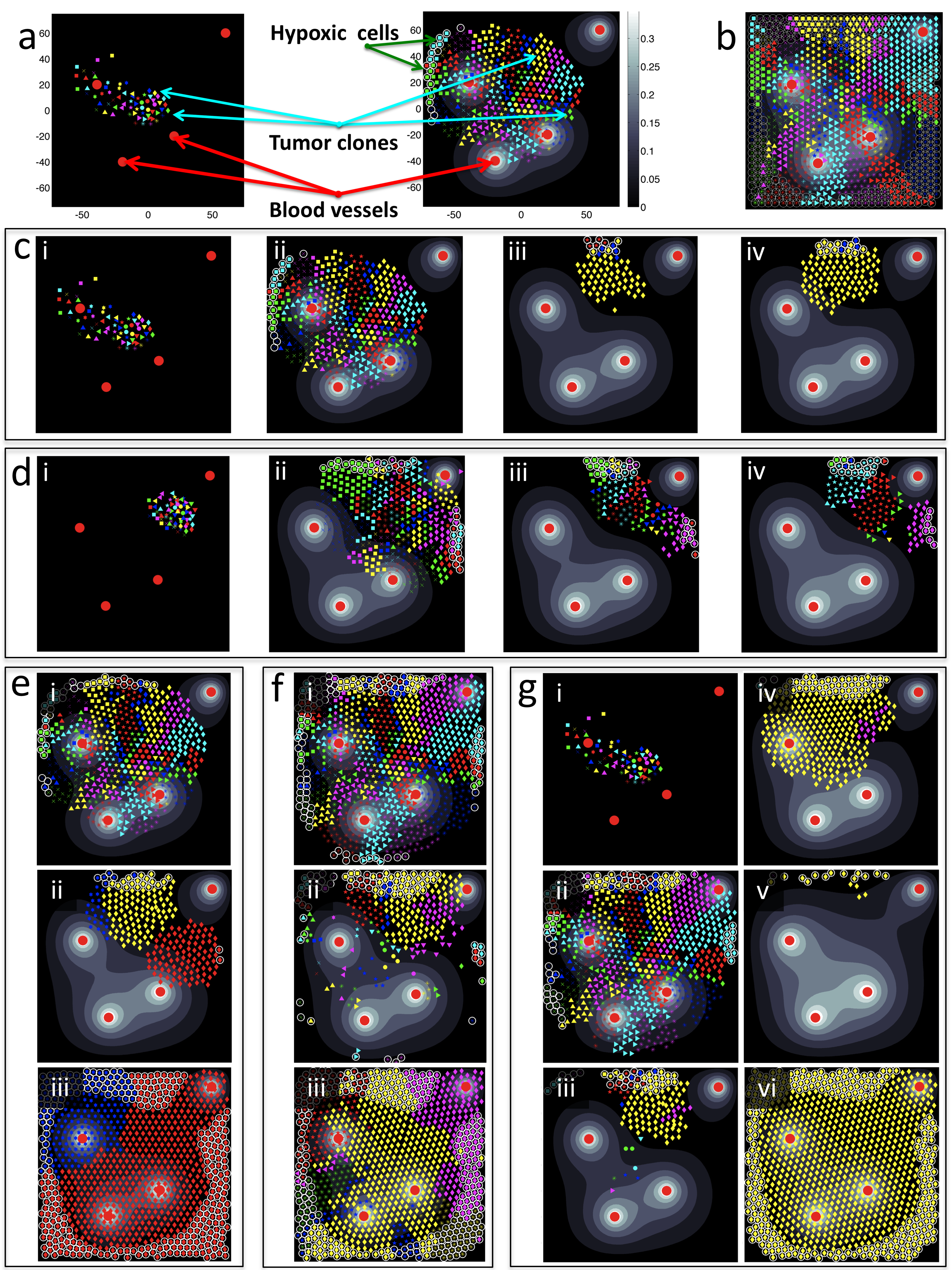}  
\caption{\label{fig:FigColor} \small{Computational model assumptions and selected simulation results. (a) Visualization of model components at the initial stage (left), and during simulation (right): the gradient of drug (greyscale in background) is regulated by influx from the vasculature (red circles) and by cellular uptake; the clonal origin of each individual cell is indicated by a different symbol/color combination; hypoxic cells are surrounded by white circles. (b) Final (after 25,000 iterations) configuration of tumor and its clonal heterogeneity in a case of no treatment. (c-d) Evolution of a tumor under treatment, but with no resistance for two distinct initial configurations of 65 cells (Section \ref{sec:NoRes}). Both tumors will be eradicated by treatment (not shown). (e) Evolution of a tumor when pre-existing resistant cells give rise to a resistant tumor (reduced effectiveness of drug during the course of treatment; see Section \ref{sec:preexistResSimul}). (f) Evolution of a tumor when cells that acquire resistance result in the formation of a resistant tumor (Section \ref{sec:acquiredRes}). (g) Evolution of tumor clones influenced by the microenvironmental niche in the case of acquired resistance (Section \ref{sec:Niches}).}}
\end{center}
\end{figure}
\clearpage

\subsection{Reaction-diffusion equation for oxygen kinetics.}\label{sec:OxyEqu}
The change in oxygen concentration $\xi$ at location $\mathbf{x}=(x,y)$ in the tumor tissue depends on oxygen supply  from the vasculature $V_j$ at a constant rate $S_{\xi}$, its diffusion with diffusion coefficient $\mathcal{D_{\xi}}$, and cellular uptake by the nearby tumor cells $C_k$ at rate $\rho_{\xi}$. It is governed by the following equation:
\begin{equation}
\displaystyle\frac{\partial \xi(\mathbf{x},t)}{\partial t}=\underbrace{\mathcal{D_{\xi}}\Delta \xi(\mathbf{x},t)}_{diffusion} - \underbrace{\min\left(\xi(\mathbf{x},t), \rho_{\xi}\displaystyle\sum_k \chi_{C_k}(\mathbf{x},t)\right)}_{uptake \;\; by \;\; the \;\; cells} + \underbrace{S_{\xi} \displaystyle\sum_j\chi_{V_j}(\mathbf{x},t)}_{supply},
\end{equation}
where $k$ indexes over the number of cancer cells with positions $C_k^{(X,Y)}$, and $j$ indexes over the number of vessels with positions $V_j^{(X,Y)}$. $\chi$ is the characteristic function defining the cell and vessel neighborhood, whose definition depends on a fixed cell radius $R_C$ and a fixed vessel radius $R_V$:
\[ \chi_{C_k}(\mathbf{x},t) = \left\{  \begin{array}{ll} 1 & \mbox{ if } \left|\left| \mathbf{x}-C_k^{(X,Y)}(t)\right|\right| < R_C\\  0 & \mbox{ otherwise,}  \end{array}\right.    \;\;\;\;\;\;\;\;\;\;\;\;\;        
   \chi_{V_j}(\mathbf{x},t) = \left\{  \begin{array}{ll} 1 & \mbox{ if } \left|\left| \mathbf{x}-V_j^{(X,Y)}\right|\right| < R_V\\  0 & \mbox{ otherwise.}  \end{array}\right. \]
   
Sink-like boundary conditions ($\partial \xi (\mathbf{x},t) / \partial \mathbf{n}=-\varpi  \xi (\mathbf{x},t)$) are imposed along all domain boundaries $\mathbf{x}~\in~\partial \Omega$, where $\mathbf{n}$ is the inward pointing normal. The initial oxygen concentration $\xi(\mathbf{x},t_0)$ in the whole model domain $\Omega$ is shown in the Appendix in {\sc Fig.}\ref{fig:FigAppendix}a. The method for determining the initial oxygen distribution, oxygen boundary conditions, and oxygen uptake rates are also described in the Appendix.

\subsection{Reaction-diffusion equation for drug kinetics.}
The change in drug concentration $\gamma$ in the tumor tissue depends on its supply from the vasculature $V_j$, its diffusion with diffusion coefficient $\mathcal{D_{\gamma}}$, decay with decay rate $d_{\gamma}$, and cellular uptake by the tumor cells $C_k$ at rate $\rho_{\gamma}$.  It is governed by the following equation:
\begin{equation}
\displaystyle\frac{\partial \gamma(\mathbf{x},t)}{\partial t}=\underbrace{\mathcal{D_{\gamma}}\Delta \gamma(\mathbf{x},t)}_{diffusion} -\underbrace {d_{\gamma} \gamma(\mathbf{x},t)}_{decay}-\underbrace{\min\left(\gamma(\mathbf{x},t), \rho_{\gamma}\displaystyle\sum_k \chi_{C_k}(\mathbf{x},t)\right)}_{uptake \;\; by \;\; the \;\; cells}+\underbrace{S_{\gamma}(t)\displaystyle\sum_j \chi_{V_j}(\mathbf{x},t)}_{supply},
\end{equation}
where $S_{\gamma}(t)$ is a time-dependent supply function from each vessel. In this paper we assume that the drug supply is continuous, thus we use a constant $S_{\gamma}$.  In principle, however, this approach allows us to model different time dependent drug administration schedules. The neighborhood function $\chi$ is defined above, and the drug boundary conditions are defined in the same way as for oxygen. 
The initial conditions are defined as follows: $\gamma(\mathbf{x},t_0) = 0$ for $\mathbf{x} \in \Omega \setminus \bigcup V_k$, and $\gamma(\mathbf{x},t_0) = S_{\gamma}$ at all $V_k^{(X,Y)}$.  Together, these equations state that $t_0$ represents the start of treatment at which time drug is only found at the sites of blood vessels.  The method for determining drug boundary conditions and drug uptake rates are described in the Appendix.

\subsection{Agent-based model for tumor cell dynamics.}
Each cell in our model is treated as a separate entity characterized by several individually regulated properties that define cell position $C^{(X,Y)}$, current cell age $C^{age}$, cell maturation age $C^{mat}$, the level of sensed oxygen $C^{\xi}$, the level of accumulated drug $C^{\gamma}$, the time of cell exposure to high drug concentration $C^{exp}$, the level of accumulated DNA damage $C^{dam}$, the level of damage that the individual cell can withstand, but upon crossing it, the cell will die (a ``death threshold") $C^{death}$, and two variables used to identify cell heritage--the unique index of the host cell ($ID_{c}$), and the unique index of its mother cell ($ID_{m}$), that we denote by $C^{(ID_{c},ID_{m})}$. The state of the $k$-th cell at time $t$, $C_k(t)$, will be denoted as follows:
\[ 
C_k(t) = \left\{ C_k^{(X,Y)}(t), C_k^{age}(t), C_k^{mat}, C_k^{\xi}(t), C_k^{\gamma}(t), C_k^{exp}(t), C_k^{dam}(t), C_k^{death}(t), C_k^{(ID_{c},ID_{m})}\right\}. 
\]
The rules for updating each cell property are described below and summarized in a flowchart of cell response to microenvironmental conditions shown in {\sc Fig. }\ref{fig:FigChart}.  Note that cells are updated in a random order at each iteration to avoid any configurational biases.

\begin{figure}[ht] 
\begin{center}
\includegraphics[width=0.75\textwidth]{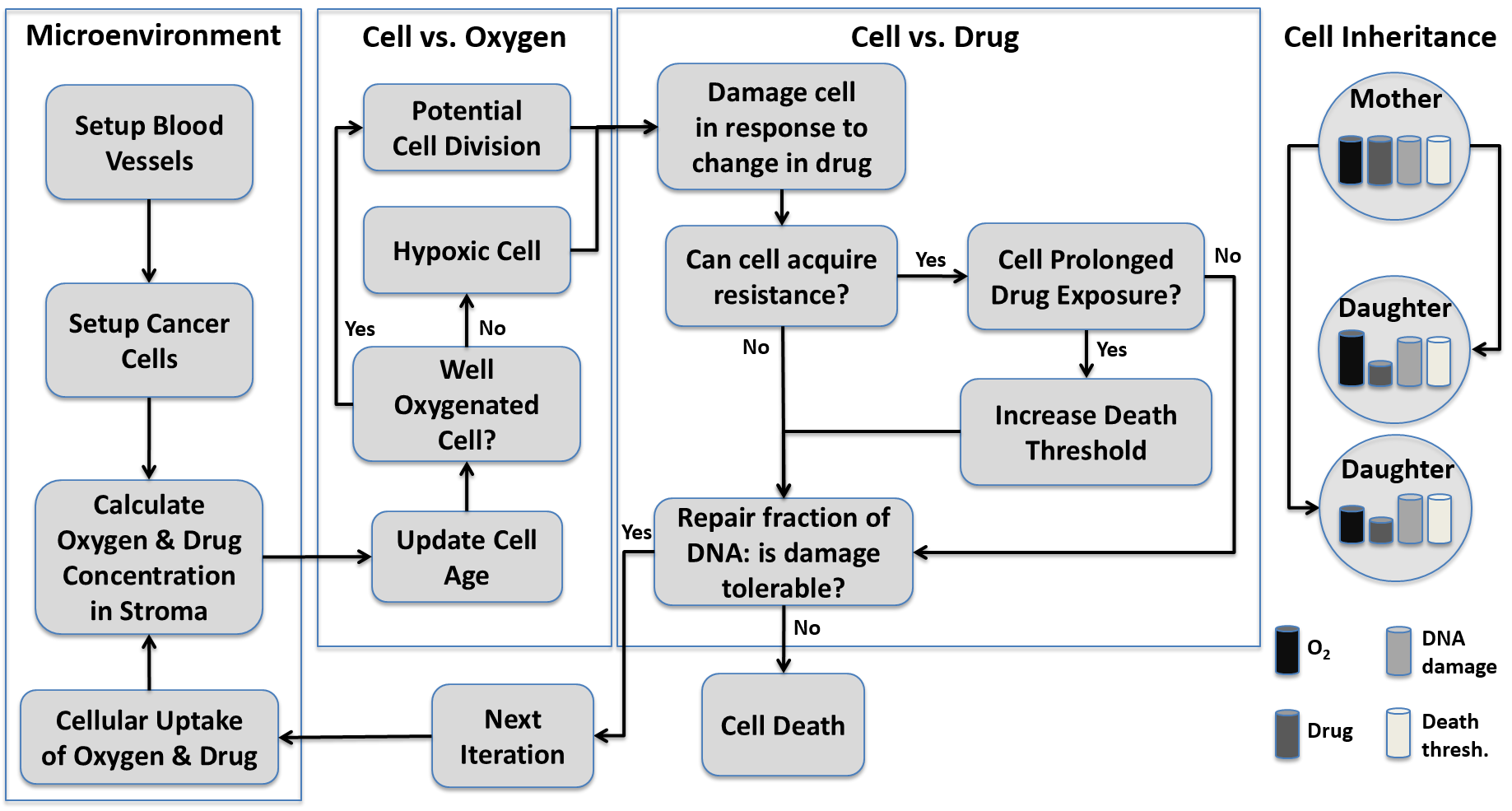}
\caption{\label{fig:FigChart}\small{Flowchart of cell behavior in response to microenvironmental factors. Based on signals sensed from the microenvironment (first column) each cell can respond to oxygen levels (second column) by potentially proliferating or becoming hypoxic, and to drug levels (third column) by either surviving, adapting, or dying. Upon cell division (rightmost column), some cell properties are inherited, while others are sampled from the environment.}}
\end{center}
\end{figure}

Each cell can inspect its local neighborhood and sense extracellular concentrations of both oxygen $\xi$ and drug $\gamma$. The quantity of oxygen taken up and used by the $k$-th cell ($C_k^{\xi}$), and the amount of drug taken up by the cell ($C_k^{\gamma}$) are determined as follows:
\[
C_k^{\xi}(t+\Delta t) =  \underbrace{\sum_{\mathbf{x}} \xi(\mathbf{x},t)}_{sensed \;\; \& \;\; used}  \mbox{and } \;\;\;
C_k^{\gamma}(t+\Delta t) = C_k^{\gamma}(t) +   \left[\max \left(0,  \sum_{\mathbf{x}} \underbrace{\min\left( \gamma(\mathbf{x},t),\rho_{\gamma}\right)}_{uptake} - \underbrace{d_{\gamma} C_k^{\gamma}(t)}_{decay}\right)\right]\Delta t,
\]
where $\left\{ \mathbf{x} : \left|\left| \mathbf{x} - C_k^{(X,Y)}\right|\right|\ < R_C \right\}$ is a local cell neighborhood. The level of sensed oxygen will regulate cell proliferative capabilities, and cells will become quiescent in a hypoxic environment (the level of sensed oxygen falls below a predefined threshold $Thr_{hypo}$). 

The duration of cell exposure to the drug and its concentration will determine cell DNA damage.  In the case of acquired resistance, the duration of drug exposure also determines the cell death threshold. While these mechanisms are assumed here, their effects are consistent with experimentally hypothesized mechanisms of resistance, including the accumulation of mutations of drug targets or inactivation of the drug \cite{DeanEtAl:2005}. The drug-induced DNA damage is assumed to depend on the current increase in drug consumed by the cell (drug uptake minus drug decay) and on DNA repair (proportional to the current damage with the rate $p$): 
\begin{equation}
C_k^{dam}(t+\Delta t) = C_k^{dam}(t) + \left[\max \left(0,  \sum_{\mathbf{x}} \min( \gamma(\mathbf{x},t),\rho_{\gamma}) - d_{\gamma} C_k^{\gamma}(t) \right)\right]  \Delta t  \; - p \; C_k^{dam}(t).
\end{equation}
If this accumulated damage exceeds the tolerated damage level (the death threshold $C_k^{death}$, determined by self-calibration in the Appendix), this cell dies.  

The damage level tolerated by the cells is defined differently depending on whether we deal with pre-existing or acquired cell resistance. In the case of pre-existing drug resistance, $C_k^{death}$ is fixed for all cellular clones, however, it is set up to be relatively higher for the resistant clones. For simplicity, we assume that it is a multiplier of the value for sensitive clones.  Thus $C_k^{death}$=$Thr_{death}$ for sensitive cells and $C_k^{death}$=$Thr_{multi} \times Thr_{death}$ for resistant cells. In the case of acquired cell resistance, the $C_k^{death}$ may increase in each cell independently (with the increase step $\Delta_{death}$) if the prolonged drug exposure criterion is met:
\[
C_k^{death}(t+\Delta t) = 
\left\{ \begin{array}{ll} 
C_k^{death}(t) + \Delta_{death} & \mbox{ if } C_k^{exp}(t) > t_{exp} \\
C_k^{death}(t) & \mbox{ otherwise}\\
\end{array}\right.
\]
\noindent where $t_{exp}$ is the prolonged drug exposure threshold that determines whether a cell's death threshold increases.  Note that $C_k^{exp}$ is defined by introducing cell short-term memory to count how long the cell has been exposed to a high drug concentration $\gamma_{exp}$: 
\[
C_k^{exp}(t+\Delta t) = 
\left\{ \begin{array}{ll} 
C_k^{exp}(t) + \Delta t & \mbox{ if } C_k^{\gamma}(t) > \gamma_{exp} \\
0 & \mbox{ otherwise}.\\
\end{array}\right. 
\]

Cell current age is updated as follows: $C_k^{age}(t+\Delta t)$=$C_k^{age}(t)$+$\Delta t$. We assume that when the cell reaches its maturation age, $C_k^{mat}$, it is ready to divide unless it is surrounded by other cells (the overcrowding condition), in which case cell proliferation is suppressed until space becomes available. Upon division of the $k$-th cell, $C_k(t)$, two daughter cells are created instantaneously ($C_{k1}(t)$ and $C_{k2}(t)$). One daughter cell takes the coordinates of the mother cell, whereas the second daughter cell is placed randomly near the mother cell; that is:
\[
C_{k1}^{(X,Y)}(t)=C_k^{(X,Y)}(t) \;\;\;\;\; \mbox{and} \;\;\;\;\; C_{k2}^{(X,Y)}(t)=C_k^{(X,Y)}(t)+R_C\left(cos(\theta),sin(\theta)\right), 
\]
where $\theta$ is a random angle. The current age of each daughter cell is initialized to zero, $C_{k1}^{age}(t)$=$C_{k2}^{age}(t)$=0, however the cell maturation age is inherited from its mother cell with a small noise term: $C_{k1}^{mat},C_{k2}^{mat}$=$C_{k}^{mat}\pm \omega$, with $\omega\in [0,C_k^{mat}/20]$. Moreover, both daughter cells inherit from their mother the level of DNA damage $C_{k1}^{dam}(t)$=$C_{k2}^{dam}(t)$=$C_{k}^{dam}(t)$, the death threshold $C_{k1}^{death}(t)$=$C_{k2}^{death}(t)$=$C_{k}^{death}(t)$, and the drug exposure time $C_{k1}^{exp}(t)$=$C_{k2}^{exp}(t)$=$C_{k}^{exp}(t)$. However, the level of accumulated drug is divided equally between the two daughter cells $C_{k1}^{\gamma}(t)$=$C_{k2}^{\gamma}(t)$=0.5$\times C_{k}^{\gamma}(t)$, and the levels of sensed oxygen $C_{k1}^{\xi}(t)$ and $C_{k2}^{\xi}(t)$ are determined independently for each cell based on oxygen contents in the cell vicinity. Furthermore, the unique index of each daughter cell consists of the newly assigned index and the inherited mother index, i.e., $C_{k1}^{(ID_c,ID_m)}$=$(k_1,C_k^{ID_c})$ and $C_{k2}^{(ID_c,ID_m)}$=$(k_2,C_k^{ID_c})$. This is summarized in the ``cell inheritance" column of {\sc Fig. }\ref{fig:FigChart}.  
Initial properties of the $k$-th tumor cells are: $\a C_k(t_0) = \left\{(X_k,Y_k), 0, M_k, \sum_{\mathbf{x}}\xi(\mathbf{x},t_0), 0, 0, 0, T_k, (k, 0)\right\}$, where $M_k$ is drawn from a uniform distribution $[0.5\times \mathcal{A}ge,1.5\times \mathcal{A}ge]$, and $\mathcal{A}ge$ is the average maturation age. $T_k$=$Thr_{death}$ for all cells in the acquired resistance case, and for all sensitive cells in the pre-existing resistance case. $T_k$=$Thr_{multi}\times Thr_{death}$ for all resistance cells in the pre-existing resistance case. This means that all cells start at their pre-defined position $(X_k,Y_k)$, at age zero with whatever amount of oxygen they sense from the local microenvironment.  There is no drug in any of the cells initially, and cells have been exposed to the drug for no time. Further, cells have not accumulated any damage and all cells have a pre-defined death threshold. Each initial cell has an unknown mother, thus its unique index is $(k,0)$, and the first component is passed along to its daughters as a mother cell index.

\subsection{Equations of cell mechanics.}\label{sec:CellMech}

The equations of cell mechanics are based on the previously published model by Meineke {\it et al.} \cite{MainekeEtAl:2001}, in which cells are represented by the coordinates of its nucleus $C^{(X,Y)}(t)$ and a fixed cell radius $R_C$. Thus, each cell is assumed to have a default volume that is maintained during its lifetime, and the neighboring cells push each other away by exerting repulsive forces if they come into contact (i.e., if the distance between the cells is less than the cell diameter $2R_C$). Therefore, the change in cell position depends on interactions with other cells in the neighborhood.
Two repulsive (linear, Hookean) forces $f_{i,j}$ and $-f_{i,j}$ will be applied to cells $C_i^{(X,Y)}(t)$ and $C_j^{(X,Y)}(t)$, respectively, to move these two cells apart and preserve cell volume exclusivity. To simplify notation in the equations below, we let ${\mathbf X}_i$=$C_i^{(X,Y)}(t)$. Then,

\[ f_{i,j} = \left\{ \begin{array}{ll} {\mathcal F}\; (2R_C - \| {\mathbf X}_i - {\mathbf X}_j \|) \frac{{\mathbf X}_i - {\mathbf X}_j}{\| {\mathbf X}_i - {\mathbf X}_j \|} & \mbox{ if } \| {\mathbf X}_i - {\mathbf X}_j \| < 2R_C \\  0 & \mbox{ otherwise},  \end{array} \right. \]    
where ${\mathcal F}$ is the constant spring stiffness, and the spring resting length is equal to cell diameter.
If the cell ${\mathbf X}_i$ is in a neighborhood of more than one cell (say, ${\mathbf X}_{j_1}, \ldots, {\mathbf X}_{j_M}$), the total force $F_i$ acting on ${\mathbf X}_i$ is the sum of all repulsive forces $f_{i,j_1}, \ldots, f_{i,j_M}$ coming from the springs between all neighboring cells connected to ${\mathbf X}_i$:
\[ 
\begin{array}{ccl} 
F_i &=& \underbrace{{\mathcal F} (2R_C - \| {\mathbf X}_i - {\mathbf X}_{j_1} \|) \frac{{\mathbf X}_i - {\mathbf X}_{j_1}}{\| {\mathbf X}_i - {\mathbf X}_{j_1} \|}}_{f_{i,j_1}}   +   \ldots   +   \underbrace{{\mathcal F}\; (2R_C - \| {\mathbf X}_i - {\mathbf X}_{j_M} \|) \frac{{\mathbf X}_i - {\mathbf X}_{j_M}}{\| {\mathbf X}_i - {\mathbf X}_{j_M} \|}}_{f_{i,j_M}}. 
 \end{array} 
 \]

Cell dynamics are governed by the Newtonian equations of motion where the connecting springs are overdamped (system returns to equilibrium without oscillations). Damping force is related linearly to velocity with a damping coefficient $\nu$, and thus the force and cell relocation equations are given by:  
\[
F_i = -\nu \frac{d{\mathbf X_i}}{dt} \;\;\;\;\; \mbox{and} \;\;\;\;\; \mathbf{X}_i(t+\Delta t) = \mathbf{X}_i(t) - \frac{1}{\nu} \Delta t F_i.
\]

When a dividing cell gives rise to two daughter cells, the repulsive forces between daughter cells are activated since they are placed at the distance smaller than cell diameter. Further, this may also result in daughter cell placement near other tumor cells.  Therefore multiple repulsive forces will be applied until the cells are pushed away and the whole tumor cluster reaches an equilibrium configuration. The introduction of repulsive cell-cell interactions also allows for monitoring the number of neighboring cells, which determines whether the cells are overcrowded. This, in turn, influences subsequent tumor growth. It is of note that some cells get pushed out of the domain after the cluster reaches an equilibrium configuration.  Only cells that remain inside the domain are considered in our analysis. All model parameters that define cell mechanics are summarized in Table \ref{tab:table1}. 
%
\begin{table}[ht]
\begin{center}
\caption{\label{tab:table1}\small{Physical and computational parameters of cell mechanics.}}
\begin{small}
\begin{tabular}{ll|ll}
\hline
\multicolumn{2}{c}{\bf cellular parameters:} & \multicolumn{2}{c}{\bf numerical parameters:} \\
Cell radius         & $R_C$=5 $\mu$m & Domain size & $[-75,75]\times[-75,75]$ $\mu$m\\
Spring stiffness & $\mathcal{F}$=1  $mg/s^2$   & Mesh width  & $h_b$=2$\mu$m  \\
Mass viscosity  & $\nu$=15  $mg/s$                &Time step      & $\Delta t$=0.5min\\
Overcrowding   &  $Neighbors$=14 cells   &&\\
Maturation age & $\mathcal{A}ge$=360min $\pm$ noise$\in$[0,$\mathcal{A}ge$/20] &&\\
\hline
\end{tabular}
\end{small}
\end{center}
\end{table}

\subsection{Numerical implementation of model equations}
A finite difference scheme is used to approximate the solution of the partial differential equations for oxygen and drug.  Space is discretized into a  square grid, with spacing between grid points of $h_b$. Time is also discretized with a time step $\Delta t$. The numerical values of these parameters are given in Table \ref{tab:table1}.
The solution to the partial differential equation is then approximated using a forward-difference approximation (in time) on a square grid (centered in space), using sink-like boundary conditions as detailed in Section \ref{sec:OxyEqu}. At each time step, the initial conditions for approximating $\a \gamma(\mathbf{x},t+\Delta t)$ and $\a \xi(\mathbf{x},t+\Delta t)$ are the values of $\a \gamma(\mathbf{x},t)$ and $\a \xi(\mathbf{x},t)$, respectively.  We prescribed the initial values, $\a \gamma(\mathbf{x},t_0)$ and $\a \xi(\mathbf{x},t_0)$, as detailed above and in Table \ref{tab:table2}, whereas the values of $\gamma$ and $\xi$ at all $t>0$ are determined by finite difference approximation.

\begin{table}[ht]
\begin{center}
\caption{\label{tab:table2}\small{Non-dimensionalized parameters of oxygen and drug kinetics based on calibration presented in Appendix.}}
\begin{small}
\begin{tabular}{lll|lll}
\hline
\multicolumn{3}{c}{\bf metabolite kinetics (normalized values):} & \multicolumn{3}{c}{\bf drug resistance (normalized values):} \\
& {\bf oxygen:} & {\bf drug:} &  & {\bf pre-existing:} & {\bf acquired:} \\
Supply rate & $S_{\xi}$=1 & $S_{\gamma}$=1 & Death thresh. incr. & $\Delta_{death}$=0 & $\Delta_{death}$ varies \\
Diffusion coefficient & $\mathcal{D_{\xi}}$=0.5  & $\mathcal{D_{\gamma}}$=0.5 & Death thresh. mult. & $Thr_{multi}$=5 & $Thr_{multi}$=1 \\
Decay rate & none & $d_{\gamma}$=1$\times 10^{-4}$ & Drug exposure level & N/A & $\gamma_{exp}$=0.01\\
Boundary outflux rate & $\varpi$=0.45 & $\varpi$=0.45 & Drug exposure time  & N/A & $t_{exp}$=5$\Delta t$\\
Cellular uptake & $\rho_{\xi}$=5$S_{\xi}$$\times$$10^{-5}$  & $\rho_{\gamma}$=$S_{\gamma}$$\times$$10^{-4}$ & DNA repair & $p$ varies & $p$=1.5$\times$$10^{-4}$ \\
Threshold value & $Thr_{hypo}$=0.05 & $Thr_{death}$=0.5 & & & \\
\hline
\end{tabular}
\end{small}
\end{center}
\end{table}

Cell-drug and cell-oxygen interactions are also modeled at the same discretized time points for which we numerically approximate the solution to the partial differential equation. In particular, at each discrete time point, any discrete grid point within a fixed distance $R_C$ from the center of a cancer cell can uptake drug for that cancer cell at a uptake rate $\rho_{\gamma}$, provided this amount of drug is available. If there is not enough drug at the site to take up $\rho_{\gamma}$ units, all the drug available at the site is taken up by the grid point for the associated cell.  Finally, we also assume that drug found inside the cell decays at the same rate as it does in the environment. In a similar way oxygen is taken from the environment at the rate $\rho_{\xi}$ and immediately used by the cells. All non-dimensionalized parameter values used in our model are summarized in Table \ref{tab:table2} (see Appendix for parameter self-calibration).

\section{Results}
An analysis of our calibrated model is undertaken here.  To facilitate comparisons between simulations so that only the impact of resistance type (none, pre-existing, acquired) is assessed, all simulations (with one noted exception) use the same initial condition for the configuration of 65 tumor cells (each individually represented by a unique shape/color combination) and blood vessels, shown in {\sc Fig.}\ref{fig:FigColor}a, left. Further, all simulations are run using a fixed seed for the random number generator. We first explain tumor dynamics upon treatment with a continuously-administered DNA damaging drug when neither mode of resistance is incorporated.  Following this, we make some analytical predictions on tumor behavior, under the simplifying assumption that long-term behavior is not influenced by the spatial features of the model.  In the case of pre-existing resistance, the theoretical analysis is supported by numerical simulations.  In the case of acquired resistance, numerical simulations are also conducted, although the connection to the theoretical analysis is less straightforward.  We conclude by studying the impact that spatial location and microenvironmental niche have on tumor response to treatment.

\subsection{Analysis of treatment outcome with no tumor resistance.}\label{sec:NoRes}
Here, we consider the case when no form of anti-cancer drug resistance can develop. This means that all cells initially have the same tolerance to drug damage (parameter $Thr_{multi}$ is set to 1), and this tolerance will not increase during the course of treatment (parameter $\Delta_{death}$ set to 0). The model has been calibrated (see Appendix) so that in this case the drug is successful and all tumor cells will be eliminated in a relatively short time (after several cell cycles). Moreover, the drug should eradicate the tumor regardless of initial cell configuration. Thus here we discuss two cases that differ only by the initial locations of tumor cells. The first case, shown in {\sc Fig.}\ref{fig:FigColor}c(i), is comprised of a dispersed cluster of cells initially located at varying distances from the vasculature. The second case is shown in {\sc Fig.}\ref{fig:FigColor}d(i) and contains cells located within the low drug niche (see Section \ref{sec:Niches}), relatively far from all vessels. All other model parameters are identical in both cases, including the death threshold ($Thr_{death}$=0.5) and DNA repair rate ($p$=1.5$\times$$10^{-4}$), as listed in Table~\ref{tab:table2}.

\begin{figure}[ht]
\center
\includegraphics[width=0.95\textwidth]{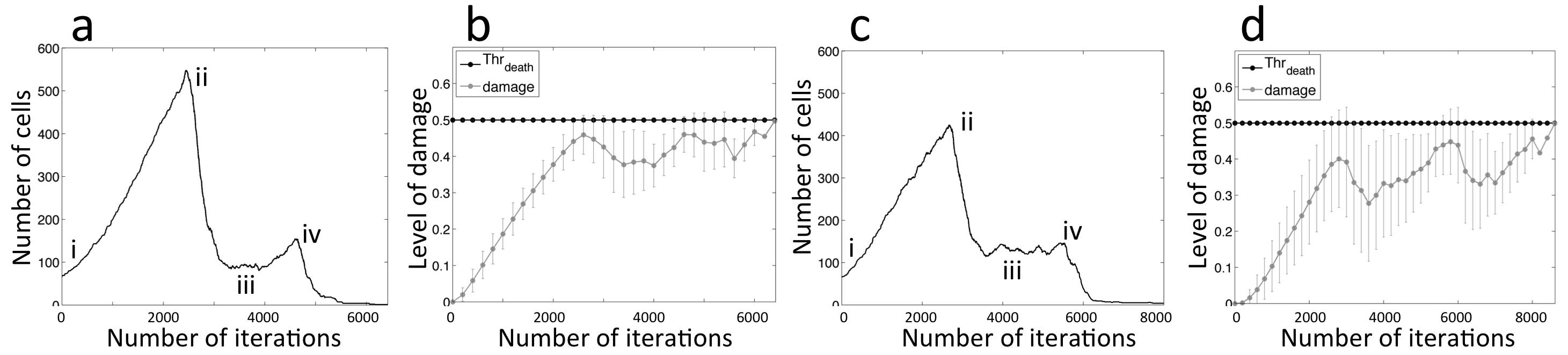} 
\caption{\label{fig:Noresistance}\small{No-resistance case. Comparison of cell evolution curves (a,c) and cell damage curves (b,d) for two distinct initial conditions shown in {\sc Fig.}\ref{fig:FigColor}c(i) and {\sc Fig.}\ref{fig:FigColor}d(i). The numbers in (a) and (c)  correspond to color panels in {\sc Fig.}\ref{fig:FigColor}c and {\sc Fig.}\ref{fig:FigColor}d, respectively. The curves in (b) and (d) show an average damage level of all cells with the standard deviation represented by vertical lines, and an average death threshold of all cancer cells, which here is constant.}}
\end{figure}

In both cases the qualitative behavior of the whole cell population is similar, as can be seen by comparing cell evolution curves in {\sc Fig.}\ref{fig:Noresistance}a and {\sc Fig.}\ref{fig:Noresistance}c. Initially, the cells do not encounter enough drug to die off, thus allowing them to increase the cell population by 4-5 fold ({\sc Fig.}\ref{fig:Noresistance}a(ii) and {\sc Fig.}\ref{fig:Noresistance}c(ii)) and overtake a large portion of the domain.  The corresponding cell configurations are shown in {\sc Fig.}\ref{fig:FigColor}c(ii) and {\sc Fig.}\ref{fig:FigColor}d(ii), respectively. During this time the majority of cells steadily accumulate DNA damage ({\sc Fig.}\ref{fig:Noresistance}b,d) and subsequently die when the damage level exceeds the death threshold ({\sc Fig.}\ref{fig:Noresistance}a(iii) and {\sc Fig.}\ref{fig:Noresistance}c(iii)).  The exception to this is a small cell cluster located in the low drug niche ({\sc Fig.}\ref{fig:FigColor}c(iii) and {\sc Fig.}\ref{fig:FigColor}d(iii)). These small subpopulations are able to briefly recover and temporarily increase tumor size ({\sc Fig.}\ref{fig:Noresistance}a(iv), {\sc Fig.}\ref{fig:FigColor}c(iv) and {\sc Fig.}\ref{fig:Noresistance}c(iv), {\sc Fig.}\ref{fig:FigColor}d(iv)), although both tumors eventually die out. Tumor eradication in both cases took place within several cell cycles, although the tumor initiated within the low drug niche survived slightly longer than the dispersed configuration: 9 cell cycles versus 11 cell cycles, with an average cell cycle counting 720 iterations. 
The main difference between these two cases is in the number of persistent clones (one in {\sc Fig.}\ref{fig:FigColor}c(iv) and several in {\sc Fig.}\ref{fig:FigColor}d(iv)), which suggests that certain cell clones may be selected by microenvironmental conditions and thus survive longer. This may have implications for the development of tumor resistance since subpopulations that confer a certain advantage within the tissue may give rise to distinct resistance capabilities or indeed to a resistant phenotype.

\subsection{Theoretical analysis of the parameter space with resistance.}\label{sec:preexistResTheory}
In this section, we theoretically analyze the damage incurred in cancer cells during treatment.  We are interested in predicting the parameter regimes for which the damage accumulated by cancer cells $C^{dam}$ exceeds their death threshold $C^{death}$, as this is the situation that results in cancer cell death.  

To facilitate our analysis, we first assume that each cell incurs the same amount of DNA damage.  In other words, at any fixed time $t$, all cells will have the same level of DNA damage.  In the equations below, we let $x(t) = C_k^{dam}(t)$ be the notation used to represent this spatially-constant damage level. Further, let $\eta(t)>0$ represent the new amount of damage incurred in each cell at time $t$, and define $0\leq p\leq 1$ to be the constant fraction of DNA damage repaired.  
Therefore, for any small $h$, and any $t>0$ we have:
\begin{eqnarray*}
x(t) = x(t-h) + \displaystyle\int_{t-h}^t \eta(s) \;ds - \displaystyle\int_{t-h}^t p x(s)\;ds.
\end{eqnarray*}
Using the Fundamental Theorem of Calculus in the limit as $h\to0$, this statement reduces to
\begin{eqnarray}\label{ODE_cell_damage}
\dot{x} &=& \eta(t) - p x(t).
\end{eqnarray}
Assuming no damage at $t$=$0$, i.e., $x(0)$=$0$, the following describes the explicit solution of Equation (\ref{ODE_cell_damage})
\begin{eqnarray}\label{ODE_cell_damage_solution}
x(t) \;=\; \left(\displaystyle\int_0^t \eta(s) e^{ps}ds \right) e^{-pt}.
\end{eqnarray}

In what follows, we further assume that $\eta(t) =\eta$ is constant, so that the amount of new damage induced by the  drug in each cell is constant.  This is not necessarily the case in our spatial model, since the damage incurred by a cancer cell depends on the amount of drug that reaches the cancer cell, and that depends on the location of the cell in tissue space, and potentially on how long the drug has been administered. However, since the drug influx from the vessels is also constant, the amount of damage in each time point is bounded from above. We call this value $\eta$.

For $p\neq0$, Equation (\ref{ODE_cell_damage_solution}) becomes:
\begin{eqnarray}\label{ODE_cell_damage_solution_constant_p_eta}
x(t) \;=\;\displaystyle\frac{\eta}{p} \left(1-e^{-pt}\right). 
\end{eqnarray}

For $p=0$, meaning there is no DNA damage repair, Equation (\ref{ODE_cell_damage_solution}) becomes:
\begin{eqnarray}\label{ODE_cell_damage_solution_constant_p_eta_p_0}
x(t) =\eta\; t. 
\end{eqnarray}

\subsubsection{Pre-existing case}
In the pre-existing case, we start with two clones with a death threshold that is larger than the rest of the cells. We call this threshold $\beta_1$=$Thr_{multi} \times Thr_{death}$.  Since these clones can tolerate more DNA damage than the others, they are called resistant. The rest of the clones are called sensitive and their death threshold is $\beta_2$=$Thr_{death}$. Moreover,  $\Delta_{death} = 0$ since no drug resistance can be acquired during the course of treatment. Therefore, $C^{death}(t)=\beta_1$ for resistant clones, and $C^{death}(t)=\beta_2$ for sensitive clones, with $\beta_2<\beta_1.$ In the following proposition we provide conditions for cancer cell eradication ($x(t)>C^{death}(t)$) and for cancer cell survival ($x(t)<C^{death}(t)$). 

\begin{proposition}\label{proposition1}
\begin{enumerate}
\item For $0\leq p < \displaystyle\frac{\eta}{\beta_1}$, both resistant and sensitive clones die, i.e., there exists $T>0$ such that for any $t>T$, $x(t)>\beta_1>\beta_2$. 

\item For $\displaystyle\frac{\eta}{\beta_1}< p < \displaystyle\frac{\eta}{\beta_2}$, resistant clones survive but sensitive clones die, i.e., there exists $T>0$ such that $x(t)>\beta_2$ for any $t>T$ but $x(t)<\beta_1$ for any $t$.

\item  For $p > \displaystyle\frac{\eta}{\beta_2}$, all clones survive, i.e., $x(t)<\beta_2<\beta_1$ for any $t$.
\end{enumerate}
\end{proposition}

\begin{proof}
For $i=1,2$, let $F_i(t):= x(t)-\beta_i$. We consider the following two cases:
\begin{itemize}
\item  $p=0$. In this case, using Equation (\ref{ODE_cell_damage_solution_constant_p_eta_p_0}), $x(t)=\eta t$. Hence, when $t> \frac{\beta_1}{\eta}$, 
$\eta t=x(t) > \beta_1 > \beta_2$. 
\item $p> 0$. In this case, using Equation (\ref{ODE_cell_damage_solution_constant_p_eta}), 
\[F_i(t)= \frac{\eta}{p}-\frac{\eta}{p} e^{-pt}-\beta_i.\]
Note that $F(0)=-\beta<0$ and $F'(t)=\eta e^{-pt}>0$. Therefore, if $\displaystyle\lim_{t\to\infty}F_i(t)=  \frac{\eta}{p}-\beta_i>0$, it follows by Intermediate Value Theorem that $F_i$ has a unique positive root $T$ and, furthermore, $F_i(t)>0$ for $t>T$. If on the other hand $\displaystyle\lim_{t\to\infty}F_i(t)=  \frac{\eta}{p}-\beta_i <0$, then $F_i(t)<0$ for all $t\geq0$.
\begin{enumerate}
\item For $p < \displaystyle\frac{\eta}{\beta_1} < \displaystyle\frac{\eta}{\beta_2}$, we have:
\[\displaystyle\lim_{t\to\infty}F_i(t)=  \frac{\eta}{p}-\beta_i>0, \quad \mbox{for $i=1,2$.}\]
Therefore, there exists $T>0$ such that $x(t)>\beta_1>\beta_2$ for all $t>T$. 

\item For $\displaystyle\frac{\eta}{\beta_1}< p < \displaystyle\frac{\eta}{\beta_2}$, we have: 
\[\displaystyle\lim_{t\to\infty}F_2(t)=  \frac{\eta}{p}-\beta_2>0,\quad \displaystyle\lim_{t\to\infty}F_1(t)=  \frac{\eta}{p}-\beta_1<0. \]
Therefore, there exists $T>0$ such that for $t>T$, $x(t)>\beta_2$, but for all $t\geq0$, $x(t)<\beta_1$. 

\item  For $p > \displaystyle\frac{\eta}{\beta_2}$, we have:
\[\displaystyle\lim_{t\to\infty}F_i(t)=  \frac{\eta}{p}-\beta_i<0, \quad \mbox{for $i=1,2$.}\]
Therefore, $x(t)<\beta_2<\beta_1$, for all $t\geq0$.
\end{enumerate}
\end{itemize}
\end{proof}

When $0 \leq p \leq \frac{\eta}{\beta}$, our analysis predicts tumor eradication by a continuously-administered DNA damaging drug.  Simulations of our model for the calibrated parameter values reveal that $\a \eta \leq \eta_{max} = 3\times 10^{-4}$. Since we have assumed a constant value of $\eta$ for all cells for all time, we will use $\eta = 3\times 10^{-4}$ to make numerical predictions from our theoretical results.  Further, we have fixed $\a \beta_1 = 2.5$ which means that part 1 of Proposition \ref{proposition1} predicts that if $\a p < 1.2\times 10^{-4}$, complete tumor eradication is expected as long as treatment is administered for a sufficiently long period of time.  

When $\frac{\eta}{\beta_1} < p < \frac{\eta}{\beta_2}$, which numerically corresponds to DNA damage repair in the range $(1.2\times 10^{-4}, 6\times 10^{-4})$, part 2 of Proposition \ref{proposition1} predicts that all sensitive clones should be eradicated, although resistant clones should persist.  Finally, when $p>6\times 10^{-4}$, all clones are predicted to survive the treatment protocol.  In Section \ref{sec:preexistResSimul}, we will demonstrate that these theoretical predictions are consistent with our numerical simulations of the full spatial model, when considering the long-term dynamics.  The theoretical analysis does not reveal information about the transient dynamics of the model.

\subsubsection{Acquired case}
We also performed a theoretical analysis of the damage level and death threshold in the case of acquired resistance.  Just as with the pre-existing analysis, we assumed that all cancer cells incur constant levels of damage for all time points during treatment ($\eta$=3$\times 10^{-4}$).  However, damage incurred by a cell actually depends on the amount of drug the cell receives, and this depends on the location of the cell in tissue space.  While this neglect of the spatial component gave informative predictions in the case of pre-existing resistance, the predictions end up being less informative in the case of acquired resistance.  For this reason, the theoretical analysis of acquired resistance will be presented in Section \ref{sec:locationAcqu}, where we discuss the importance of spatial location and microenvironmental niche on tumor survival.

\subsection{Simulations in the case of pre-existing resistance.}\label{sec:preexistResSimul}
In the case of pre-existing resistance, a subpopulation of tumor cells that are resistant to the DNA damaging drug are present at the initiation of treatment.  As detailed in Section \ref{sec:preexistResTheory}, pre-existing resistance is simulated through the $Thr_{multi}$ parameter.  While sensitive clones have a death threshold $Thr_{death}$=0.5, the threshold of resistant cells is increased by a multiple of $Thr_{multi} > 1$.  We randomly choose two cells (approximately 3\% of the initial tumor population) to have pre-existing resistance.  Since the location and proximity of a resistant cell to a blood vessel could have a major impact on whether the cell will live or die, here we only consider the two resistant cells to be at an intermediate distance from the vessels.  In Section \ref{sec:locationPre}, we will explore the impact of the location of the resistant cells.

In the pre-existing resistance simulations, we vary $p$, the constant fraction of DNA damage repaired.  Although $Thr_{multi}$ is the more obvious pre-existing resistance parameter, it is the relationship between $p$ and $Thr_{multi}$ that determines tumor response to the DNA damaging drug.  The time scales in the model are easier to control (and this facilitates comparison between different cases) when we fix $Thr_{multi}$ and allow $p$ to vary.  For all pre-existing analysis and simulations, we fix $Thr_{multi} = 5$.

\begin{figure}[ht]
\begin{center}
\includegraphics[width=0.95\textwidth]{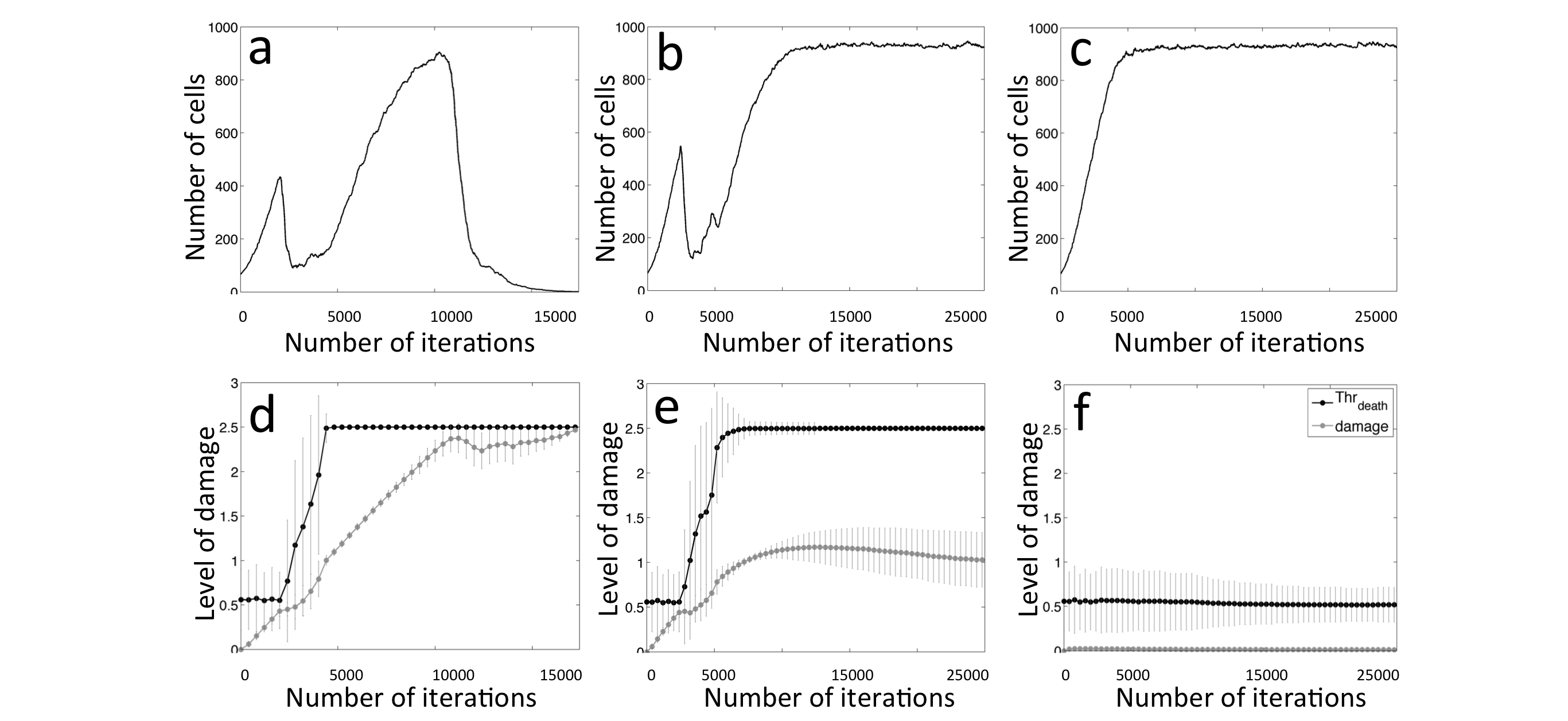}
\caption{\label{fig:FigPreexist}\small{Pre-existing resistance case. A comparison of tumor evolution curves (a-c) and tumor damage curves (d-f) in three distinct cases: tumor eradication when $p=0$ (a,d), tumor resistance when $p=1.5\times 10^{-4}$ (b,e), and treatment failure when $p=10^{-2}$ (c,f). The curves in (d-f) show an average damage level and an average death threshold level of all cells with the standard deviation represented by vertical lines.}} 
\end{center}
\end{figure}

\subsubsection{Successful treatment}\label{sec:preexistingRes_success}
Consistent with our theoretical analysis, continuous administration of a DNA damaging drug results in eventual tumor eradication when $0\leq p \leq$1.2$\times$$10^{-4}$.  The case of $p$=0 is shown in {\sc Fig. }\ref{fig:FigPreexist}a,d.  Within this parameter regime, tumor behavior is qualitatively similar to the case of no drug resistance ({\sc Fig.}\ref{fig:Noresistance}a,b).  Both simulations exhibit a transient period for which cancer cells accumulate, but do not yet respond to, the drug.  During this time period, a clonally heterogeneous tumor begins to grow in tissue space while, simultaneously, the average cellular DNA damage level is increasing.  Since the average death threshold in the cancer cell population is low (the majority of cells starts with a small death threshold), the sensitive cells start dying and the tumor population shrinks in size ({\sc Fig.}\ref{fig:FigPreexist}a,d). This leads to selection of the resistant cancer cells, as demonstrated by the rapid increase in the average death threshold in {\sc Fig.}\ref{fig:FigPreexist}d.  This tumor, eventually composed entirely of resistant cells (saturation in the average death threshold), undergoes a significant period of regrowth before the entire population experiences a rapid increase in DNA damage levels and is eradicated by the drug. Therefore the model reveals that a DNA damaging drug can eliminate a tumor with a subpopulation of resistant cells that existed prior to the onset of treatment.

\subsubsection{Emergence of a drug resistant tumor}\label{sec:preexistingRes_resist}
Our theoretical analysis predicts that when $1.2\times 10^{-4}<p<6\times 10^{-4}$, resistant cells will survive treatment while sensitive cells will be eradicated.  The dynamics of this process, however, can only be revealed through numerical simulations.  We find that in this parameter regime, the effectiveness of the DNA damaging drug is reduced during the course of treatment ({\sc Fig.}\ref{fig:FigPreexist}b where $p=1.5\times 10^{-4}$).  This is what we consider to be a drug resistant tumor, and what is often observed clinically when treating patients with a chemotherapeutic agent.

In this parameter regime, the early tumor dynamics are comparable to the successful treatment case: after an initial period of time during which cancer cells build up drug levels, there is a transient period of tumor shrinkage ({\sc Fig.}\ref{fig:FigPreexist}b). {\sc Fig.}\ref{fig:FigColor}e(ii) reveals that only three clones are selected for during this period of time: the two resistant clones (red and blue), and the one sensitive clone located in a low-drug niche (yellow). These surviving clones begin to repopulate tissue space.  Unlike in the successful treatment case, the drug cannot eliminate the resistant cells, though it can eliminate the sensitive cell in the low drug niche.  This is seen in {\sc Fig.}\ref{fig:FigPreexist}e where the average death threshold of the surviving cells has stabilized at a value of 2.5, whereas the average damage level is actually decreasing from a maximum value of approximately 1.25. Therefore the damage levels cannot surpass the death threshold of the cells, and at the end of the simulation period, we are left with a resistant tumor composed entirely of resistant cells.

\subsubsection{Complete treatment failure}\label{sec:preexistingRes_failure}
Consistent with theoretical predictions, when $p > 6\times 10^{-4}$, all clones survive treatment with a DNA damaging drug.  Simulations further reveal that the tumor grows monotonically in this parameter regime, meaning the drug fails to exhibit any anti-tumor activity ({\sc Fig.}\ref{fig:FigPreexist}c where $p=10^{-2}$).  Treatment failure occurs because the level of DNA damage repair is so large it prevents the damage level of each cancer cell, whether resistant or sensitive, from exceeding its death threshold ({\sc Fig.}\ref{fig:FigPreexist}f).  In this parameter regime, a clonally diverse tumor overtakes space, comparable to what is observed when the tumor grows with no treatment ({\sc Fig.}\ref{fig:FigColor}b).

\subsection{Computational analysis of the parameter space in the case of acquired resistance.}\label{sec:acquiredRes}
In the case of acquired resistance, a subpopulation of tumor cells evolve a resistant phenotype due to selective pressures imposed by the drug, microenvironment, or other factors.  In other words, no drug resistant cells are found in the tumor when treatment is initiated.

As detailed in Section \ref{sec:MathModel}, acquired resistance is simulated in an individual cell in response to prolonged drug exposure.  If a cell meets the prolonged drug-exposure criterion, then its death threshold $C_{k}^{death}$ gets incremented by $\Delta_{death}$.  A cell whose death threshold gets increased in this way can tolerate greater levels of DNA damage, as cell death is triggered when the damage level of a cell $C_{k}^{dam}$ exceeds its death threshold.  

We computationally analyzed the model fixing all parameters as detailed in Table \ref{tab:table2}. This includes fixing the DNA repair constant (the parameter varied in the pre-existing resistance case) to $p=1.5\times 10^{-4}$.  This $p$ was chosen because it gave rise to a drug resistant tumor in the case of pre-existing resistance.  In order to isolate the impact of acquired drug resistance, the only parameter that was varied here is $\Delta_{death}$.  Computational analysis of the model reveals four distinct parameter regimes, depending on the value of $\Delta_{death}$: 1) successful treatment; 2) almost successful treatment; 3) eventual emergence of a drug resistance tumor; 4) complete treatment failure. 

\begin{figure}[ht]
\begin{center}
\includegraphics[width=0.95\textwidth]{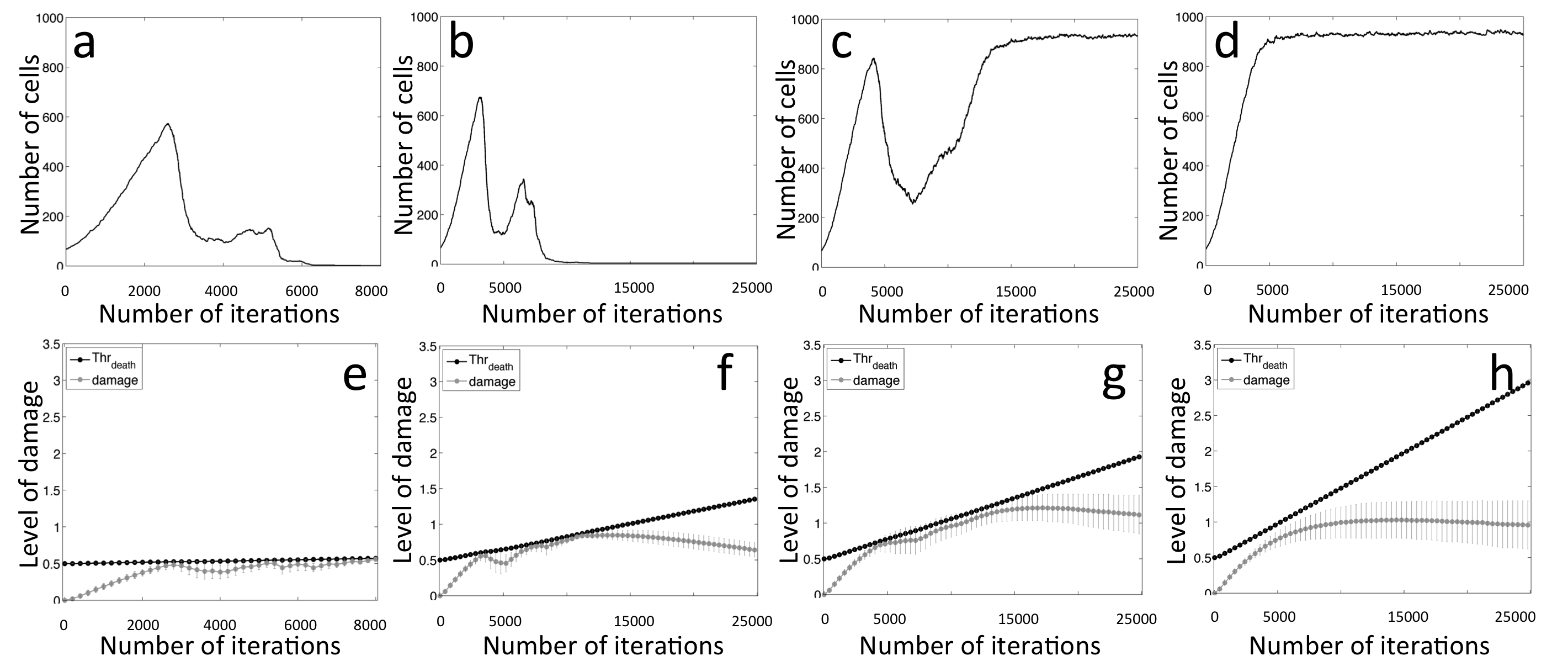}
\caption{\label{fig:FigAcquired}\small{Acquired resistance case. A comparison of tumor evolution curves (a-d) and tumor damage curves (e-h) in four cases: tumor eradication when $\Delta_{death}$=$10^{-5}$ (a,e), almost successful treatment when $\Delta_{death}$=3.5$\times 10^{-5}$ (b,f), tumor resistance when $\Delta_{death}$=5.9$\times 10^{-5}$ (c,g), and treatment failure when $\Delta_{death}$=$10^{-4}$ (d,h). The curves in (e-h) show an average damage level and an average death threshold with the standard deviation represented by vertical lines.}}
\end{center}
\end{figure}

\subsubsection{Successful treatment}\label{sec:acquiredRes_success}

As in the case of pre-existing resistance, there is a parameter regime for which treatment is successful in spite of any resistance acquired to the DNA damaging drug.  Simulations reveal that the tumor eradication  parameter regime is: $0\leq \Delta_{death}< 3\times 10^{-5}$. The case of $\Delta_{death} = 10^{-5}$ is illustrated in {\sc Fig.}\ref{fig:FigAcquired}a,e.  Within this parameter regime, the tumor is not monotonically decreasing in size, as observed in the no resistance case and in the successful treatment of a tumor with pre-existing resistant cells case ({\sc Fig.}\ref{fig:Noresistance}a).  There is a transient period of time (approximately five cell cycles in length) during which the cancer cells accumulate drug but do not respond to the drug.  

Once drug levels have caused enough DNA damage, a transient period of tumor shrinkage occurs.  The cancer cells that are eliminated during this period are the ones with adequate access to the drug.  This period of tumor shrinkage is always followed by a period of regrowth, during which cells in the low-drug microenvironmental niche temporarily repopulate tissue space.  However, this period of drug-induced resistance is transient, and the tumor is eventually eradicated by the DNA damaging drug.  

Therefore, the model reveals that a cancer can be successfully eliminated by a DNA damaging drug that  induces resistance after prolonged exposure.  In other words, just because the cells have the ability to  acquire resistance to the drug, this does not mean a drug resistant tumor will develop.

\subsubsection{Almost successful treatment}\label{sec:acquiredRes_almostSuccess}

Simulations also revealed a parameter regime that was not found in the pre-existing resistance case.  When $\Delta_{death}$ satisfies $3\times 10^{-5} \leq \Delta_{death} < 4\times 10^{-5}$, the drug shrinks the tumor to a very small number of cancer cells ({\sc Fig.}\ref{fig:FigAcquired}b where $\Delta_{death} = 3.5\times 10^{-5}$).  The reason the tumor is not fully eradicated by treatment is because the surviving cells are all stuck in hypoxic microenvironmental niches (not shown). Evidence that the few cancer cells found in the hypoxic niche will survive treatment is seen in {\sc Fig.}\ref{fig:FigAcquired}f.  The average death threshold of the surviving cells is increasing (cells are continuing to acquire resistance) while the average damage level is decreasing (cells are repairing more damage than they are accumulating).  

Other than the persistence of a small number of hypoxic cancer cells, the dynamics observed in this parameter regime are qualitatively the same as the successful treatment parameter regime. Further, since the vasculature is not evolving in our model and since the cells are not actively motile, there is no mechanism for the surviving hypoxic cells to re-enter the cell cycle and trigger tumor expansion.  Therefore, this is another parameter regime for which cells have the ability to acquire resistance to a DNA damaging drug, but a drug resistant tumor does not develop.

\subsubsection{Emergence of a drug resistant tumor}\label{sec:acquiredRes_resistance}

Computational analysis revealed that for $4\times 10^{-5}\leq \Delta_{death} < 7\times 10^{-5}$, the effectiveness of the DNA damaging drug is reduced over time.  Just as in the pre-existing case, this is what we consider a drug resistant tumor.  

The precise trajectory of the tumor, and its clonal composition, depends sensitively on the selected value of $\Delta_{death}$ within this parameter range.  For instance, when $\Delta_{death} = 4.5\times 10^{-5}$, there are two transient periods of tumor shrinkage before a clonally homogeneous tumor overtakes tissue space ({\sc Fig.}\ref{fig:FigColor}g).  The one surviving tumor clone is initially located in low-drug niche in tissue space, and eventually emerges from these lower-drug regions to overtake space. This will be discussed further in Section \ref{sec:locationAcqu}. 

On the other hand, when $\Delta_{death} = 5.9\times 10^{-5}$, there is only one transient period of tumor shrinkage ({\sc Fig.}\ref{fig:FigAcquired}c) preceding the formation of a drug resistant tumor.  The spatial dynamics and clonal history of the tumor are interesting to track in this case.  The growth curve in {\sc Fig.}\ref{fig:FigAcquired}c shows a local maximum around 4000 iterations.  At this point a clonally heterogeneous tumor has practically overtaken tissue space ({\sc Fig.}\ref{fig:FigColor}f(i)).  However, {\sc Fig.}\ref{fig:FigAcquired}g shows that around this local maximum, the average damage level of the cancer cells catches up to the average death threshold, and this leads to a period of time during which the drug effectively reduces the tumor size.  At the local minimum, some clonal diversity has been lost.  Most cells are found in either hypoxic or low-drug niches, although there are some cells found in normoxic zones where drug is readily accessible ({\sc Fig.}\ref{fig:FigColor}f(ii)). 

After the tumor escapes its local minimum, the average death level of the surviving cancer cells diverges from the average damage level, and a more clonally heterogeneous tumor overtakes tissue space ({\sc Fig.}\ref{fig:FigColor}f(iii)).  It is interesting to compare this to the case of pre-existing resistance.  The number of clones that survive in pre-existing drug resistant tumors is precisely the number of resistant clones present in the tumor before any treatment.  In the acquired case, the number of clones in the resistant tumor sensitively depends on the value of $\Delta_{death}$ within the drug resistant parameter range.  We observed that when $\Delta_{death} = 4.5\times 10^{-5}$ only one clone survives, which is less clonally diverse than our pre-existing resistance case. $\Delta_{death} = 5.9\times 10^{-5}$ corresponds to nine different clones surviving treatment, which is more clonally diverse than the pre-existing case ({\sc Fig.}\ref{fig:FigColor}e(iii)).

\subsubsection{Complete treatment failure}\label{sec:acquiredRes_failure}
Computational analysis reveals that when $\Delta_{death} \geq 7\times 10^{-5}$, the DNA damaging drug fails to cause any reduction in the number of cancer cells.  Treatment failure occurs because the drug induces resistance in the cancer cells quicker than it damages the DNA.  In this parameter regime, a clonally diverse tumor overtakes tissue space, and the tumor never shrinks in response to treatment ({\sc Fig.}\ref{fig:FigAcquired}d,h where $\Delta_{death} = 10^{-4}$).  The clonal composition of the tumor after 25,000 iterations of growth is comparable to the composition of a tumor grown with no treatment ({\sc Fig.}\ref{fig:FigColor}b).

\subsection{Impact of spatial location and microenvironmental niches on tumor survival.}\label{sec:Niches}
Microenvironmental conditions such as the position of vasculature can influence tumor growth and heterogeneity. Cells require oxygen to survive and replicate, but proximity to blood vessels also means higher exposure to drug. Regions of the tissue landscape where blood vessels are far enough away to minimize drug exposure while still close enough to maintain adequate oxygen levels will facilitate tumor growth. This can already be seen in the non-resistant tumor cases (Section \ref{sec:NoRes}). While both discussed tumors died out, it took longer to eradicate the tumor initiated near the microenvironmental niche of low drug concentration. Moreover, this niche also gave rise to a subpopulation of cells able to recover from drug insult, although only temporarily.

Our tissue landscape contains a cluster of three blood vessels and an additional solitary blood vessel near the boundary of the tissue space ({\sc Fig.}\ref{fig:5_1}a). This leads to a low-drug/normoxic niche near the center of the tissue sample, and a low-drug/hypoxic region at much of the boundary of the region.

\begin{figure}[ht] 
\begin{center}
\includegraphics[width=0.9\textwidth]{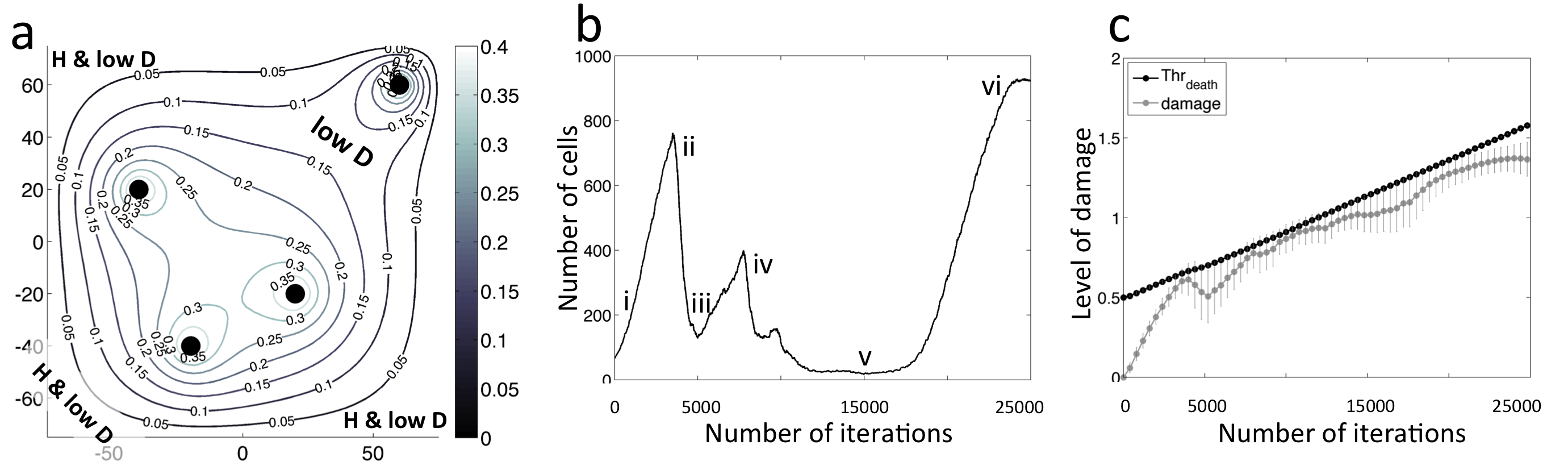}
\end{center}
\caption{\label{fig:5_1}\small{Evolution of acquired resistance as an effect of microenvironmental niches of low drug and hypoxia. (a) Tissue landscape with two microenvironmetal niches: a low-drug/hypoxic (H \& low D) and low-drug/normoxic (low D). The grey-scale contours indicate drug distribution. The black circles indicate four vessels. (b) A cell population evolution curve with numbers corresponding to color panels in {\sc Fig.}\ref{fig:FigColor}g. (c) An average damage level and an average death threshold of all cells with the standard deviation represented by vertical lines.}}
\end{figure}

\subsubsection{Pre-existing resistance: impact of cell location}\label{sec:locationPre}
As discussed in Section \ref{sec:preexistResSimul}, the initial condition for our pre-existing resistance simulations was a large population of sensitive cells with a subpopulation of two resistant cells (approximately 3\% of the population). To study how the microenvironment influences the survival and heterogeneity of tumors under this condition, we varied which two cells were chosen to be resistant. In particular, we tested three different configurations: 1) resistant cells chosen to be the two cells closest to blood vessels; 2) resistant cells chosen at an intermediate distance from the blood vessels; and 3) resistant cells chosen to be the two cells furthest from the vessels.

Consistent with our theoretical analysis of pre-existing resistance in Section \ref{sec:preexistResTheory}, we found that long-term tumor behavior was controlled by the relationship between the value of $p$ (the fraction of damage repaired) and the pre-exiting resistance parameter $Thr_{multi}$.  It was only minimally influenced in the short term by the positioning of the two resistant cells in space. Our simulations showed that the same three parameter regimes discussed in Section \ref{sec:preexistResSimul} led to tumor eradication, a transient period of tumor shrinkage followed by treatment failure, and monotonic tumor growth respectively, regardless of the proximity of the two resistant cells to the vasculature ({\sc Fig.}\ref{fig:5_2}).

We conclude that in the case of pre-existing resistance, just the presence of the resistant cells combined with the DNA repair parameter value drives the overall tumor dynamics.  The position of resistant cells, on the other hand, only serves to adjust the time scale of tumor growth or death. This suggests that pre-existing resistance is a stronger promoter of tumor growth than is a preferential microenvironment.

\begin{figure}[ht] 
\begin{center}
\includegraphics[width=0.7\textwidth]{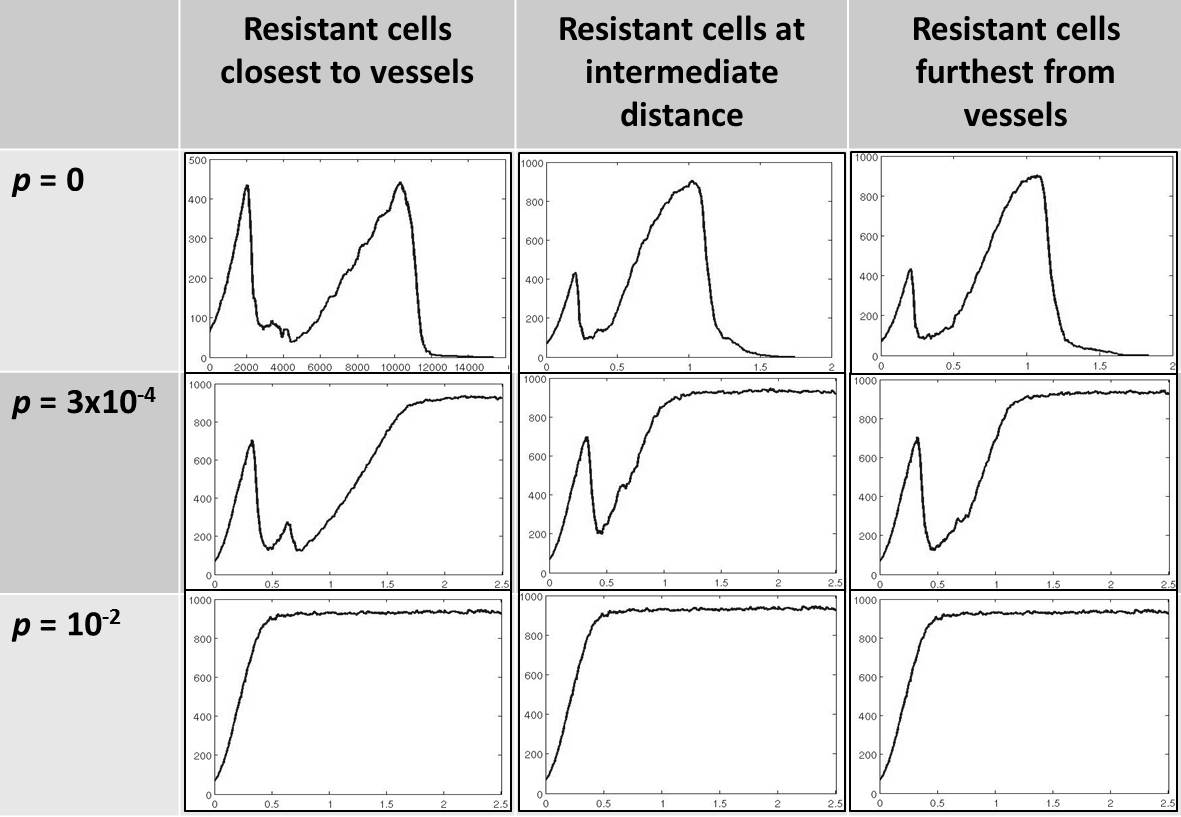}
\end{center}
\caption{\label{fig:5_2}\small{Evolution of tumor size for different choices of resistant cell location and values of DNA damage repair~$p$. For each plot shown, the $x$-axis is the number of iterations (in units of $10^3$) and the $y$-axis is the number of cells.}}
\end{figure}

\subsubsection{Acquired resistance: impact of microenvironmental niche}\label{sec:locationAcqu}
To illustrate the important role of the spatial component of the model in the case of acquired resistance, we first turn to a theoretical analysis that neglects the spatial component.  Consider $\alpha= \Delta_{death} \neq 0$, which corresponds to cells having a death threshold of approximately 
$C^{death}(t)= \alpha t +\beta$, where $\beta = Thr_{death} = 0.5$ is fixed.  As in Section~\ref{sec:preexistResTheory}, let $\eta>0$ represent the new amount of damage incurred in each cell at time $t$ (assumed to be constant), and let  $p$ be the constant fraction of damage repaired.  In the following proposition we provide a condition that guarantees that if we ignore the spatial component of the model, the cancer should be eradicated by treatment (that is, $x(t) = C_k^{dam}(t)>C^{death}(t)$).  

\begin{proposition}\label{proposition2} 
If $\alpha<\eta$ and $0<p<\displaystyle\frac{\eta}{\beta}$, then there exists $\alpha^*<\alpha$, $T_1(\alpha^*)$ and $T_2(\alpha^*)$ such that 
for all $T_1(\alpha^*)<t<T_2(\alpha^*)$, $x(t)>\alpha t+\beta$. (All cells die.)
\end{proposition}

\begin{proof}

Let $F(t):= x(t)-\alpha t-\beta$.  Using Equation (\ref{ODE_cell_damage_solution_constant_p_eta}), 
\[F(t)= \frac{\eta}{p}-\beta-\alpha t-\frac{\eta}{p} e^{-pt}.\]
 Note that $F(0)=-\beta <0$ and $F'(t)= \eta e^{-pt} -\alpha$. 
Since $\alpha<\eta$, 
\[
F'(t)=\left\{\begin{array}{ccc}\geq0 &  & t\leq t^* \\0 &  & t=t^* \\\leq0 &  & t\geq t^*\end{array}\right.
\]
where $t^*=-\displaystyle\frac{1}{p} \ln\frac{\alpha}{\eta}.$
Therefore, the graph of $F$ is one of the three graphs shown in Fig. \ref{graph:F}, 
depending on the size of $\alpha$. 
Note that if for some $\alpha<\eta$, $F(t^*)>0$, then there exists an interval $(T_1, T_2)$, such that for $t\in (T_1, T_2)$, $F(t)>0$, i.e., $x(t)> \alpha t+\beta$. Therefore, we will look for those $\alpha$'s, $\alpha<\eta$, such that $F(t^*)>0$.

\begin{figure}[ht]
\begin{center}
 \includegraphics[width=0.8\textwidth]{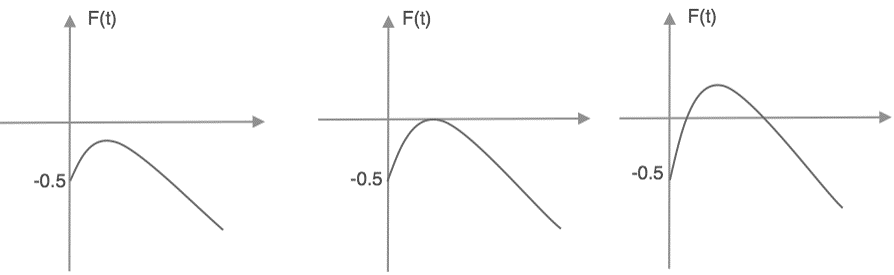}
\caption{\label{graph:F}\small{All possible graphs for function $F$, depending on the value of $\alpha.$}}
\end{center}
\end{figure}

Let $F(t^*)= G(\alpha) =: \displaystyle\frac{\eta}{p}-\beta-\displaystyle\frac{\alpha}{p}\left(1-\ln\frac{\alpha}{\eta}\right).$ 
Observe that:
\begin{itemize}
\item $\displaystyle\lim_{\alpha\to0^+} G(\alpha) =  \displaystyle\frac{\eta}{p}-\beta= \left\{\begin{array}{ccc}>0 &  & p<{\eta}/{\beta} \\<0 &  & p>{\eta}/{\beta}\end{array}\right.$.
\item $\displaystyle\lim_{\alpha\to\eta^-} G(\alpha) =-\beta<0.$
\item  on $(0,\eta)$, $G'(\alpha) = \displaystyle\frac{1}{p} \ln\displaystyle\frac{\alpha}{\eta}<0$. 
\end{itemize}
Therefore, if $p<\displaystyle\frac{\eta}{\beta}$, by Intermediate Value Theorem, there exists $0<\alpha^*<\eta$ such that $G(\alpha^*)= 0$, and $G(\alpha)>0$ for $\alpha<\alpha^*$, and $G(\alpha)<0$ for $\alpha>\alpha^*$. 
Hence, if $p<\displaystyle\frac{\eta}{\beta}$, and $\alpha<\alpha^*$, then $F(t^*)>0$, and there exist $T_1$, $T_2$, depending on $\alpha^*$, such that for $T_1< t< T_2$, $F(t)=x(t)-\alpha t -\beta >0$. 

\end{proof}

Proposition \ref{proposition2} says that if $\a p < \frac{\eta}{0.5} < 6\times 10^{-4}$ and if $\a \alpha < \eta < 3\times 10^{-4}$ (both upper bounds fixed using parameters in our model), all cancer cells should be eradicated by the DNA damaging drug.  In our numerical simulations, we consider $\a p = 1.5\times 10^{-4}$, so $p$ satisfies this constraint.  Further, the maximum value of $\alpha$ we consider is $10^{-4}$, which also satisfies the constraint imposed by this proposition.  Therefore the proposition implies that for all the simulations presented in Section \ref{sec:acquiredRes}, treatment should successfully eradicate all tumor cells.  Yet, this is  only consistent with two of the four numerical simulations presented.  The other two cases correspond to a drug resistant tumor and complete treatment failure.  

This discrepancy can be explained by the one assumption made during the theoretical analysis: we assumed that all cancer cells incur constant levels of damage at each time point during treatment ($\eta$).  However, the damage incurred by a cell is actually a function of both space and time.  So, if we consider the statement from the proposition that $p < \eta/0.5$ guarantees tumor eradication, and if we insert the fact that we use $p =1.5\times 10^{-4}$ in our simulations, this implies that $\eta > 7.5\times 10^{-5}$ is needed to guarantee the tumor eradication.  

However, cells located at sufficient distances from the vasculature accumulate less drug than cells closer to the vasculature, and therefore incur less damage.  For these cells, the condition $\eta > 7.5\times 10^{-5}$ is not necessarily met, as $\eta$ will be small when drug levels are low.  This phenomenon can result in the theoretical prediction (which ignores the spatial component) failing to match up with the numerical simulations (which include the spatial component).  Thus, space is clearly playing a very important role in treatment response when resistance to a DNA damaging drug is acquired during the course of treatment.

To further understand the importance of the microenvironmental niche in the case of acquired resistance, we tracked the clonal evolution of the tumor. We found that the microenvironmental niche plays a bigger role in determining the growth and heterogeneity of the tumor when $\Delta_{death}$ is not high.  When $\Delta_{death}$ is sufficiently small, the microenvironmental niche influences which cells survive the longest, though all cells will eventually die. Cells originally located near the low-drug/normoxic niche initially accumulate damage at a much slower rate than their neighbors while still increasing their death threshold. The rest of the less favorably placed population die out first as damage exceeds the death threshold, but the benefit of the low-drug/normoxic niche enables the few remaining cells to survive longer and even experience a second short period of growth before damage overwhelms them (Section \ref{sec:acquiredRes_success}). 

For mid-range values of $\Delta_{death}$, we observed the case of \textit{almost successful treatment} discussed in Section \ref{sec:acquiredRes_almostSuccess}. This behavior is driven by the microenvironmental niche. Cells initially placed near the low-drug/normoxic niche are selected for, and the increase in $\Delta_{death}$ lends them the time needed to reproduce enough to move the tumor mass toward the hypoxic region. The result is a colony of a few surviving hypoxic cells descended from the parent cells initially placed in the low drug niches.

When $\Delta_{death}$ is at a higher value, the microenvironmental niche also determines which cells survive.  In this case, however, the surviving cells actually result in treatment failure. The tumor that eventually takes over tissue space lacks clonal diversity. The selective powers of the microenvironmental niche are highlighted in {\sc Fig.}\ref{fig:FigColor}g.  The first period of selection occurs after a clonally heterogeneous tumor starts to overtake tissue space.  Essentially all cells, with the exception of those in low-drug niches, accumulate high enough damage levels to be killed by the drug ({\sc Fig.}\ref{fig:FigColor}g(iii)).  The tumor temporarily repopulates some of space from these low-drug regions, but again experiences a period of die-off when cells enter higher-drug regions.  The cells that survive this selection process are mostly found in a low-drug/hypoxic niche, and it is these cells that eventually result in a clonally homogeneous tumor repopulating tissue space ({\sc Fig.}\ref{fig:FigColor}g(v,vi)).

For large values of $\Delta_{death}$, however, the role of the microenvironmental niche is reduced. The cancer cells develop resistance to the drug very quickly for sufficiently large $\Delta_{death}$.  As few clones die out, this results in a monotonically growing tumor ({\sc Fig.}\ref{fig:FigAcquired}d) composed of clones from all over the tissue landscape, similar to what is observed in the case of no treatment ({\sc Fig.}\ref{fig:FigColor}b).

In conclusion, the microenvironmental niche plays different roles for different values of $\Delta_{death}$, the rate at which cells build up tolerance to the DNA damaging drug. When it is small, the microenvironment selects which cells will survive the longest before ultimately dying. As $\Delta_{death}$ increases to mid-range values, cells located near the low-drug/normoxic niche are able to develop resistance and propagate, eventually leading to the formation of a resistant tumor. For large values of $\Delta_{death}$, location becomes less important and cell survival is guaranteed by the increase of the death threshold alone.

\section{Discussion.}\label{sec:Discussion}
In this paper we addressed the emergence of a resistant tumor when a DNA damaging drug is continuously administered by an intravenous injection. We have developed a spatial agent-based model that can examine tumor response to the drug in a heterogeneous microenvironment.  This environment is comprised of non-uniformly spaced vasculature that results in an irregular gradient of oxygen and drug. We specifically focused on the difference between pre-existing and acquired drug resistance. By using identical initial microenvironmental conditions we were able to compare the dynamics and clonal evolution of the developing tumor in these two cases. We have also conducted a theoretical analysis of the temporal components of the model for both types of resistance.

The model produced several interesting results. In both cases considered, the pre-existing or acquired drug resistance is initiated independently in individual cells (cellular level). However, during the course of treatment three types of population behaviors (tissue level) are observed: 1) tumor eradication; 2) emergence of a drug resistant tumor; and 3) a non-responsive tumor.  In order to observe these three parameter regimes in the case of pre-existing resistance, DNA damage repair must be included in the model ($p>0$).  Varying the pre-existing resistance parameter ($Thr_{multi}$) without including DNA damage repair can only give rise to the case of tumor eradication (Proposition \ref{proposition1}.1).  On the other hand, in the case of acquired resistance all three parameter regimes are achievable without the inclusion of DNA damage repair.  Final tumor behavior in this case depends on the speed with which the cells build tolerance to DNA damage ($\Delta_{death}$).

In particular, in the cases with pre-existing resistance, both the theoretical analysis and the agent-based simulations produced the same conclusions: a) with a low DNA damage repair term, both the resistant and sensitive clones died; b) with a medium DNA damage repair term, resistant clones survived but sensitive clones died; and c) for high levels of the DNA damage repair term, all clones survived regardless of their phenotype. The reason for this correlation between computational simulations and theoretical analysis is that in the pre-existing resistance case, the spatial component has no significant effect on the long-term outcome. Similarly, in the acquired resistance case: a) with a slow increase in the cell death threshold, all of the cells eventually died out; b) with an intermediate increase in the death threshold, some of the clones die but the tumor eventually develops resistance to the drug; and c) with a high increase in the cell death threshold, all clones are able to survive resulting in treatment failure. 

In the cases where a drug resistant tumor forms ({\sc Fig.}\ref{fig:FigColor}g,h), some common spatial dynamics are observed.  The cells located near vasculature are first affected by the drug and often die. The least affected are the cells that occupy low drug/normoxia niches, and these cells are often able to repopulate the tumor. Final clonal configurations can vary from mono-clonal (acquired case) to bi-clonal (pre-existing case), to multi-clonal (acquired case). However, in the case of acquired resistance, we have also observed a case when only a small number of cells remained quiescent in the hypoxic areas for a prolonged time. These cells have not accumulated enough damage to die, and since they are located in the area of both low drug and oxygen, they have time for DNA repair. Thus, the existence of microenvironmental niches of either low drug or low oxygen concentrations are the driving forces in both transient and long-term tumor cell survival. Further, in the case of acquired resistance, the spatial location of certain clones makes them more fit than the other cells and allows them to overtake the available space. In conclusion, when the cell has a pre-existing resistance the genetic fitness advantage seems to be more important than any spatial fitness advantage, whereas the spatial position of tumor clones plays a more important role when drug resistance is acquired.

When studying the case of acquired resistance we assumed that all cells react to drug exposure in the same way and are thus all equally capable of acquiring resistance. This choice allowed us to focus on how microenvironmental differences select which cells acquire resistance, but our model could also be extended to take into account intrinsic cellular differences. In particular, it is possible that acquired drug resistance occurs when only a single cell, rather than the whole population, possesses the ability to acquire resistance. Investigating the interaction of cellular and microenvironmental differences is an interesting problem, especially because our simulations show that varying how strongly cells respond to drug exposure ($\Delta_{death}$) can lead to clonally heterogeneous or homogenous tumors. Since we framed acquired resistance broadly as a reaction of cells to increased and prolonged exposure to drug, it is possible to use this general framework to account for a variety of the mechanisms of acquired resistance. For example, resistance may be due to increased drug efflux pumps or a simple lack of any drug dosage effect. Both of these mechanisms could be studied with our model by changing how we think about the damage level and death threshold. The damage level is related to how much drug the cell has encountered, and the death threshold at its most general level describes how the cell reacts to damage. Thus, it is possible that the cell reacts to damage less (and hence has a higher death threshold) because of enhanced drug efflux pumps or a less effective drug.

Our model can also be extended to address other important aspects of tumor resistance. In this paper, we limited our investigation to a continuously administered drug. In the future we will investigate different drug scheduling schemes including typical clinical protocols, as well as metronomic and adaptive schedules. Beyond the drug schedule, we can expand the model from considering DNA damaging drugs to considering other drugs with different killing mechanisms, including anti-mitotic drugs or drugs activated in specific microenvironmental conditions (low oxygen or high acidosis). This allow us to further extend our model to study drug combinations. This is especially important, because in clinical practice, when the tumor cells become resistant to a given drug, the treatment is often changed to another therapeutic agent. However, it has been observed that resistance to one drug is accompanied by resistance to other drugs whose structures and mechanisms of action may be completely different (multiple drug resistance). Thus, this poses interesting questions for future research. If the hypothesis of a pre-existing population of resistant cells is true, what mechanisms enable those cells to resist the drug action of the often multiple chemotherapeutic treatments that may be given to a patient sequentially or in parallel? If the hypothesis of gradual emergence of drug resistance is true, what factors contribute to the development of acquired drug resistance? The mathematical framework developed here has the potential to address multiple aspects of drug resistance in solid tumors and test methods for increasing efficacy of drug combinations.

Our model belongs to a category of spatial hybrid discrete-continuous models of anti-cancer drug resistance. While mathematical modeling of tumor growth and therapy (mathematical oncology) dates back now over half a century \cite{byrne}, the modeling of anti-cancer drug resistance have gained its momentum in the last couple of years \cite{KomarovaWodarz:2005,Lavi201290,brocato,Foo201410}. However, most of the mathematical models to date addressed the problem of drug resistance on a level of the whole cell population using a variety of mathematical approaches: stochastic models \cite{Komarova2007523}, evolutionary dynamics \cite{doi:10.1021/mp2002279,doi:10.1021/mp200270v,FooEtAl:2013,citeulike:12452919}, game theory \cite{1478-3975-9-6-065007}, Lamarckian induction \cite{PiscoEtAl:2013},  compartmental pharmacokinetic models \cite{HadjiandreouMitsis:2013}, or continuous PDE models \cite{JacksonByrne:2000,LeeEtAl:2013}.
Very few models of drug resistance have, like ours, considered the significant role of spatial tumor structure and/or interactions between cells and their microenvironment.
Silva and Gatenby used an agent-based model of cells equipped with internal metabolic machinery to investigate cell-cell and cell-microenvironment interactions during chemotherapy, and strategies to prolong survival in the case of pre-existing resistance \cite{SilvaGatenby:2010}. This work showed that administration of the chemotherapy with the goal of stabilization of tumor size instead of eradication would yield better results than use of maximum tolerated doses, thus preventing or at least delaying tumor relapse. The authors demonstrated that fast growing sensitive cells can serve as a shield keeping the resistant cells trapped inside the tumor.
Lorz {\it et al} used a continuous model of anti-cancer therapy resistance under the assumption that resistance is induced by adaptation to drug environmental pressures \cite{MZA:8816821}. This has been modeled using a concept of the expression level of a resistance gene influencing tumor cells birth/death rates, effects of chemotherapies and mutations. The same group has also considered  how the spatial structure plays a role in resistance development under combined therapy protocols. By including spatial structure into the model, the authors were able to suggest that adaptation to local conditions (microenvironment) is directly linked to resistance development \cite{lorz}.
Lavi {\it et al} used another continuous model of multi-drug resistance with a variable cell resistance level that takes the form of a structure population model. This allowed the authors to explore how cells evolve (and may be selected for) under stress imposed by cytotoxic drugs \cite{Lavi15122013}. The same group has also investigated how trait selection may give rise to different types of resistance and what implications this may have for tumor heterogeneity (at the level of mutations) \cite{Greene2014627}.
A recent work of Powathil {\it et al}, that uses the Cellular Potts framework has been employed to investigate how the cell-cycle dynamics and oxygen concentration changes influence the development of resistance  \cite{2014arXiv1407.0865P}. The authors suggested that cell-cycle-mediated drug resistance emerges because the chemotherapeutic treatment gives rise to a dominant, slow-cycling subpopulation of tumor cells, causing the drug failing to target all cancer cells.

Our model differs significantly from the models discussed above. We not only consider and compare two types of resistance, but also identify two microenvironmental niches that have an impact on emergence of resistant cells. While other hybrid models use similar approach to model tumor cell interactions with oxygen and drugs, they treat tumor microenvironment as a homogeneous medium \cite{SilvaGatenby:2010}. In contrast, our model incorporates a more realistic configuration of tumor vasculature that produces gradients of metabolites that are of irregular shapes  and can change dynamically in time. Thus we can directly observe the emergence of microenvironmental niches ({\sc Fig.}\ref{fig:FigColor}g(iii),(iv)) that protect cells from killing by drugs, and enable some of them to develop resistance. We also investigated the interplay between tumor clonal development within the spatially and temporally variable distributions of both drug and oxygen, that has not been addressed by any of the previous hybrid models of drug resistance. The presented model is quite general and the self-calibration methods have been used for it parameterization (see the Appendix).  

Our model reproduces a broad range of tumor behaviors observed in clinical practice. However, our chosen parameters have not been tuned to represent a particular drug and a particular tumor type. Therefore the model constitutes a good starting point for more precise calibration to both tumor morphology and drug pharmacokinetic properties.  In particular, in the future we plan to use more realistic tumor tissue morphologies based on patients' histology samples. This will include both more realistic vasculature that may vary between tumors of different origins, various stromal components such as stromal cells (fibroblasts or adipocytes) and immune cells (macrophages, T cells, lymphocytes), extracellular matrix fibril structure, and various distributions of metabolites (oxygen, glucose, acids, MMPs). 
In particular, the role of the tumor microenvironment in the development of drug resistance is becoming a key consideration in the development of novel chemotherapeutic agents. The interactions between tumor cells and their surrounding physical environment can influence cell signaling, survival, proliferative capacities and cell sensitivity to drugs. Thus extracellular factors including tumor hypoxia and acidity, as well as tumor cell density and the extracellular matrix composition that limit drug penetration, should be investigated in a quantitative way via combination of laboratory experimentaion and mathematical modeling.

\section*{Acknowledgements}
We thank Drs. Ami Radunskaya and Trachette Jackson for organizing the WhAM! Research Collaboration Workshop for Women in Applied Mathematics, Dynamical Systems and Applications to Biology and Medicine, that allowed us to initiate this research project, and the Institute for Mathematics and Its Applications (IMA) for a generous support during the WhAM! workshop and reunion meetings. We also acknowledge the German Research Foundation (Deutsche Forschungsgemeinschaft, DFG) for the travel grant (CE 243/1-1).

\appendix 

Before testing various mechanism of tumor resistance, our model has been calibrated to 1) achieve a stable gradient of oxygen when no cancer cells are present, as would be the case in healthy tissue; 2) generate a tumor cluster that completely fills the available space when no drug is applied as would take place in non-treated tumors; this will result in another stable gradient of oxygen with hypoxic areas located far from the vasculature; and 3) completely eliminate the tumor when the cells do not acquire resistance. These three self-calibration steps are discussed in this section.

\begin{figure}[ht]
\begin{center}
 \includegraphics[width=0.9\textwidth]{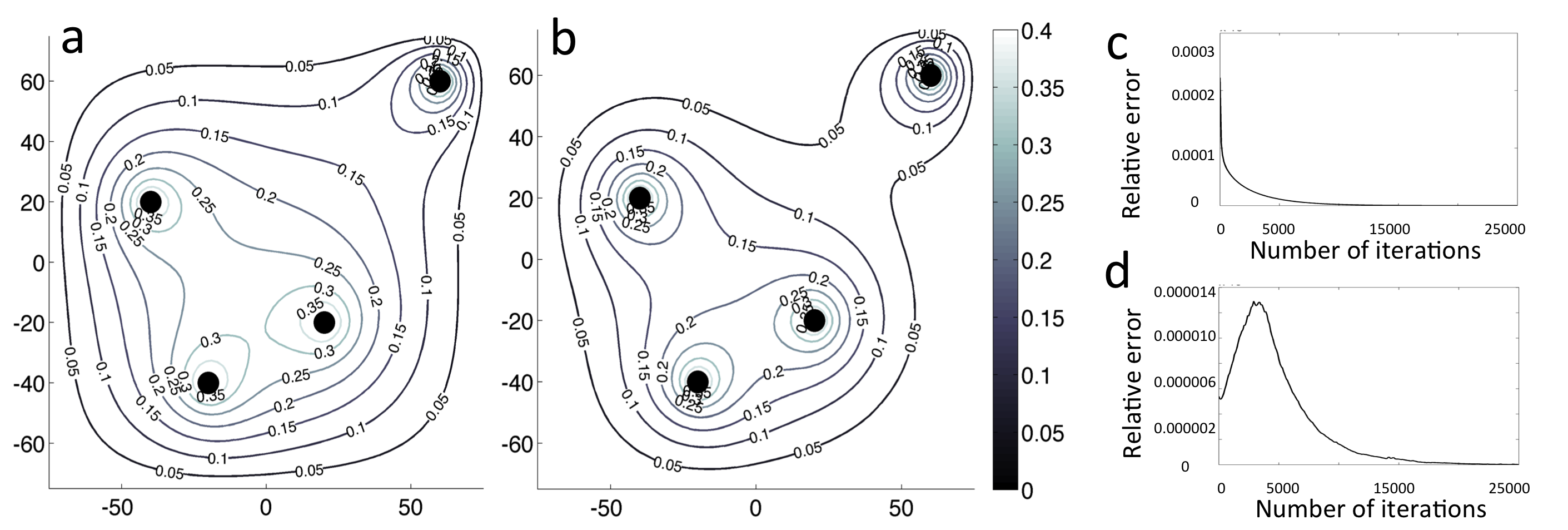}
\caption{\label{fig:FigAppendix}\small{Oxygen distribution and its stabilization error curves. (a,b) Numerically stable gradients of oxygen in a domain with no cancer cells (a), and in a domain where non-treated cancer cells uptake oxygen (b). The grey-scale contours indicate oxygen distribution levels. (c,d) Relative 2-norm error of oxygen changes $\| \eta(\mathbf{x},t+\Delta t) - \eta(\mathbf{x},t) \|_2$ calculated for 25000 iterations showing its numerical stability for cases (a) and (b), respectively.}}
\end{center}
\end{figure}

First, the influence of oxygen and drug are normalized so that $S_{\xi}$=$S_{\gamma}$=$1$.  Then we determine the oxygen diffusion coefficient $\mathcal{D_{\xi}}$ and oxygen boundary conditions that lead to a (numerically) stable gradient of oxygen when no cancer cells are present; that is with no cellular uptake. Several boundary conditions were considered, however, the best results in terms of irregular gradient stabilization and the extent of hypoxic areas were achieved for the sink-like conditions with  $\varpi$=0.45 (see Section \ref{sec:OxyEqu}). The resulting stable oxygen gradient is shown in Fig. \ref{fig:FigAppendix}a and the relative 2-norm error between two oxygen concentrations generated in two consecutive time steps is shown in Fig. \ref{fig:FigAppendix}c. The obtained oxygen gradient served as an initial condition for the reaction-diffusion equation for oxygen. 

Next, a tumor cell oxygen uptake rate $p_{\xi}$ and a hypoxia threshold $Thr_{hypo}$ were selected to allow for tumor growth with a small subpopulation of hypoxic cells, and for generation of a (numerically) stable gradient of oxygen when the tumor reaches its stable configuration. This population of tumor cells, including the hypoxic cell fraction and tumor clonal composition, serves as a control case (with no treatment) shown in Fig. \ref{fig:FigColor}b.  For the stable tumor population the (numerically) stable oxygen gradient is shown in Fig. \ref{fig:FigAppendix}b, and the relative $2$-norm error of the oxygen changes over $25,000$ iterations is shown in Fig. \ref{fig:FigAppendix}d. 

Finally, the drug diffusion coefficient $\mathcal{D_{\gamma}}$, drug uptake rate $p_{\gamma}$, and death threshold $Thr_{death}$ were determined so that the cluster of tumor cells with no resistance is eradicated. This ensures that for the chosen drug parameters, the drug is effective when there are no resistant tumor cells. This case is discussed in Section \ref{sec:NoRes}. All parameters determined by the procedure described here are listed in Table \ref{tab:table2}.
%



\begin{thebibliography}{10}

\bibitem{KimTannock:2005}
{\sc J.J. Kim}, {\sc I.F. Tannock}, \textit{Repopulation of cancer cells during therapy: an important cause of treatment failure}, Nature Reviews Cancer, \textbf{5} (2005) pp.516--525.

\bibitem{DeanEtAl:2005}
{\sc M. Dean}, {\sc T. Fojo}, {\sc S Bates}. \textit{Tumour stem cells and drug resistance.}, Nature Reviews Cancer, \textbf{5(4)} (2005) pp.275--284.

\bibitem{Baguley:2010} 
{\sc B.  Baguley}. \textit{Multiple drug resistance mechanisms in cancer.}, Molecular Biotechnology, \textbf{46(3)} (2010) pp.308--316.

\bibitem{ZahreddineBorden:2013}
{\sc H. Zahreddine}, {\sc K.L B. Borden}. \textit{Mechanisms and insights into drug resistance
in cancer}, Frontiers in Pharmacology, \textbf{4} (2013) pp.e28.

\bibitem{BockLengauerEtAl:2012}
{\sc C. Bock}, {\sc T. Lengauer}. \textit{Managing drug resistance in cancers: lessons from HIV therapy.}, Nature Reviews Cancer, \textbf{12} (2013) pp.494--501.

\bibitem{LambertEtAl:2011}
{\sc G. Lambert}, {\sc L. Est\'{e}vez-Salmeron}, {\sc S. Oh}, {\sc D. Liao}, {\sc B.M. Emerson}, {\sc T.D. Tlsty}, {\sc R.H. Austin}.  \textit{An analogy between the evolution of drug resistance in bacterial communities and malignant tissues.}, Nature Reviews Cancer, \textbf{11}:375--382 (2011).

\bibitem{TredanEtAl:2007}
{\sc O. Tr\'{e}dan}, {\sc C.M. Galmarini}, {\sc K. Patel}, {\sc I.F. Tannock}. \textit{Drug resistance
and the solid tumor microenvironment.}, Journal of the National Cancer Institute, \textbf{99(19)}:1441--1454 (2007).

\bibitem{WuEtAl:2013}
{\sc A. Wu}, {\sc K. Loutherback}, {\sc G. Lambert}, {\sc L. Est\'{e}vez-Salmeron}, {\sc T.D. Tlsty}, {\sc R.H. Austin}, {\sc J.C. Sturma}. \textit{Cell motility and drug gradients in the emergence of resistance to chemotherapy.}, PNAS, \textbf{110(40)}:16103--16108 (2013).

\bibitem{RottenbergEtAl:2007}
{\sc S. Rottenberg}, {\sc A.O. H. Nygren}, {\sc M. Pajic}, {\sc F.W. B. van Leeuwen}, {\sc I. van der Heijden}, {\sc K. van de Wetering}, {\sc X. Liu}, {\sc K.E. de Visser}, {\sc K.G. Gilhuijs}, {\sc O. van Tellingen}, {\sc J.P. Schouten}, {\sc J. Jonkers}, {\sc P. Borst}. \textit{Selective induction of chemotherapy resistance of mammary tumors in a conditional mouse model for hereditary breast cancer.}, Proceedings of the National Academy of Sciences, \textbf{104(29)} (2007) pp.12117--12122.

\bibitem{CorreiaBissell:2012}
{\sc A.L. Correia}, {\sc M.J. Bissell}. \textit{The tumor microenvironment is a dominant force in multidrug resistance.},  Drug Resistance Updates, \textbf{15}:39--49  (2012). 

\bibitem{Lage:2008}
{\sc H. Lage}. \textit{An overview of cancer multidrug resistance: a still unsolved problem.}, Cellular and Molecular Life Sciences, \textbf{65} (2008) pp.3145--3167.

\bibitem{Cheung-OngEtAl:2013}
{\sc K. Cheung-Ong}, {\sc G. Giaever}, {\sc C. Nislow}. \textit{DNA-Damaging Agents in Cancer Chemotherapy: Serendipity and Chemical Biology.}, Cell Chemistry \& Biology, \textbf{20} (2012) pp.648--659.

\bibitem{WoodsTurchi:2013}
{\sc D. Woods}, {\sc J.J. Turchi}. \textit{Chemotherapy induced DNA damage response Convergence of drugs and pathways.}, Cancer Biology \& Therapy, \textbf{14(:5)} (2012) pp.379--389.

\bibitem{MeadsEtAl:2009}
{\sc M.B. Meads}, {\sc R.A. Gatenby}, {\sc W.S. Dalton}. \textit{Environment-mediated drug resistance: a major contributor to minimal residual disease.}, Nature Reviews Cancer, \textbf{9} (2009) pp.665--674.

\bibitem{NakasoneEtAl:2012}
{\sc E.S. Nakasone}, {\sc H.A. Askautrud}, {\sc T. Kees}, {\sc J.H Park}, {\sc V. Plaks}, {\sc A.J. Ewald}, {\sc M. Fein}, {\sc M.G. Rasch}, {\sc Y.X Tan}, {\sc J. Qiu}, {\sc J. Park}, {\sc P. Sinha}, {\sc M.J. Bissell}, {\sc E. Frengen}, {\sc Z. Werb}, {\sc M. Egeblad}. \textit{Imaging Tumor-Stroma Interactions during Chemotherapy Reveals Contributions of the Microenvironment to Resistance.}, Cancer Cell, \textbf{21} (2012) pp.488--503.

\bibitem{McMillinEtAl:2013}
{\sc D.W. McMillin}, {\sc J.M. Negri}, {\sc C.S. Mitsiades}. \textit{The role of tumour-stromal interactions in modifying drug response: challenges and opportunities.}, Nature Reviews Drug Discovery, \textbf{12} (2013) pp.217--228.

\bibitem{ProvenzanoHingorani:2013}
{\sc P.P. Provenzano}, {\sc S.R. Hingorani}. \textit{Hyaluronan, fluid pressure, and stromal resistance in pancreas cancer}, British Journal of Cancer, \textbf{108} (2013) pp.1--8.

\bibitem{Karran:2001}
{\sc P. Karran}. \textit{Mechanisms of tolerance to DNA damaging therapeutic drugs.}, Carcinogenesis, \textbf{22(12)} (2001) pp.1931--1937.

\bibitem{SawickaEtAl:2004}
{\sc M. Sawicka}, {\sc M. Kalinowska}, {\sc J. Skierski}, {\sc W. Lewandowski}, \textit{A review of selected anti-tumour therapeutic agents and reasons for multidrug resistance occurrence.}, Pharmacy and Pharmacology, \textbf{56} (2004) pp.1067--1081.

\bibitem{MainekeEtAl:2001}
{\sc F.A. Meineke}, {\sc C.S. Potten}, {\sc M. Loeffler}, \textit{Cell migration and organization in the intestinal crypt using a lattice-free model.} Cell Proliferation \textbf{34} (2001) pp.253--266.

\bibitem{byrne}
{\sc H. M. Byrne.} \textit{Dissecting cancer through mathematics: from the cell to the animal
  model}, Nat Rev Cancer, \textbf{10(3)} (2010) pp. 221--230.

\bibitem{KomarovaWodarz:2005}
{\sc N. L. Komarova}, {\sc D. Wodarz}. \textit{Drug resistance in cancer: Principles of emergence and prevention}, PNAS, \textbf{102(27)} (2005) pp. 9714--9719.

\bibitem{Lavi201290}
{\sc O. Lavi}, {\sc M. M. Gottesman}, {\sc D. Levy}. \textit{The dynamics of drug resistance: A mathematical perspective.} Drug Resistance Updates, \textbf{15(1–2)} (2012) pp. 90 -- 97.

\bibitem{brocato}
{\sc T. Brocato}, {\sc P. Dogra}, {\sc E. J. Koay}, {\sc A. Day}, {\sc Y-L. Chuang}, {\sc Z. Wang}, {\sc V. Cristini}. \textit{Understanding drug resistance in breast cancer with mathematical oncology.} Current Breast Cancer Reports, \textbf{6(2)} (2014) pp. 110--120.

\bibitem{Foo201410}
{\sc J. Foo}, {\sc F. Michor}. \textit{Evolution of acquired resistance to anti-cancer therapy.} Journal of Theoretical Biology, \textbf{355} (2014) pp.10--20.

\bibitem{Komarova2007523}
{\sc N. L. Komarova}, {\sc D. Wodarz}. \textit{Stochastic modeling of cellular colonies with quiescence: An application to drug resistance in cancer.}
Theoretical Population Biology, \textbf{72(4)} (2007) pp. 523--538.

\bibitem{doi:10.1021/mp2002279}
{\sc J. J. Cunningham}, {\sc R. A. Gatenby}, {\sc J. S. Brown}. \textit{Evolutionary dynamics in cancer therapy.} Molecular Pharmaceutics, \textbf{8(6)} (2011) pp. 2094--2100.

\bibitem{doi:10.1021/mp200270v}
{\sc S. M. Mumenthaler}, {\sc J. Foo}, {\sc K. Leder}, {\sc N. C. Choi}, {\sc D. B. Agus}, {\sc W. Pao}, {\sc P. Mallick}, {\sc F. Michor}. \textit{Evolutionary modeling of combination treatment strategies to overcome resistance to tyrosine kinase inhibitors in non-small cell lung cancer.} Molecular Pharmaceutics, \textbf{8(6)} (2011) pp. 2069--2079.

\bibitem{FooEtAl:2013}
{\sc J. Foo}, {\sc K. Leder}, {\sc S. M. Mumenthaler.} \textit{cancer as a moving target: understanding the composition and rebound growth kinetics of recurrent tumors.} Evolutionary Applications, \textbf{6} (2013) pp. 54--69.  

\bibitem{citeulike:12452919}
{\sc I. Bozic}, {\sc J. G. Reiter}, {\sc B. Allen}, {\sc T. Antal}, {\sc K. Chatterjee}, {\sc P. Shah}, {\sc Y. S. Moon}, {\sc A. Yaqubie}, {\sc N. Kelly}, {\sc D. T. Le}, {\sc E. J. Lipson}, {\sc P. B. Chapman}, {\sc L. A. Diaz}, {\sc B. Vogelstein}, {\sc M. A. Nowak}. \textit {Evolutionary dynamics of cancer in response to targeted combination therapy}. eLife, \textbf{2} (2013) pp. 1-15.

\bibitem{1478-3975-9-6-065007}
{\sc P. A. Orlando}, {\sc R. A. Gatenby}, {\sc J. S. Brown}. \textit{Cancer treatment as a game: integrating evolutionary game theory into the optimal control of chemotherapy.} Physical Biology, \textbf{9(6)} (2012) pp. 065007.

\bibitem{PiscoEtAl:2013}
{\sc A. O. Pisco}, {\sc A. Brock}, {\sc J. Ahou}, {\sc A. Moor}, {\sc M. Mojtahedi}, {\sc D. Jackson}, {\sc S. Huang}. \textit{Non-Darwinian dynamics in therapy-induced cancer drug resistance}, Nature Communications, \textbf{4} (2013) pp. 2467.

\bibitem{HadjiandreouMitsis:2013}
{\sc M. M. Hadjiandreou}, {\sc G. D. Mitsis}. \textit{Mathematical modeling of tumor growth, drug-resistance, toxicity, and optimal therapy design}, IEEE Trans Biomed Eng. \textbf{61(2)} (2014) pp.415--425.

\bibitem{JacksonByrne:2000}
{\sc T. L. Jackson}, {\sc H. M. Byrne}. \textit{A mathematical model to study the effects of drug resistance and vasculature on the response of solid tumors to chemotherapy}, Mathematical biosciences, \textbf{164} (2000) pp. 17--38.

\bibitem{LeeEtAl:2013}
{\sc J. J. Lee}, {\sc J. Huang}, {\sc C. G. England}, {\sc L. R. McNally}, {\sc H. B. Frieboes}. \textit{Predictive modeling of in vivo response to gemcitabine in pancreatic cancer}, PLoS Computational Biology, \textbf{9(9)} (2013) pp. e1003231.

\bibitem{SilvaGatenby:2010}
{\sc A. S. Silva}, {\sc R. A. Gatenby}. \textit{A theoretical quantitative model for evolution of cancer chemotherapy resistance}, Biology Direct, \textbf{5(1)} (2010), pp. 25.

\bibitem{MZA:8816821}
{\sc A. Lorz}, {\sc T. Lorenzi}, {\sc M. E. Hochberg}, {\sc J. Clairambault}, {\sc B. Perthame}. \textit{Populational adaptive evolution, chemotherapeutic resistance and multiple anti-cancer therapies}, ESAIM: Mathematical Modelling and Numerical Analysis, \textbf{47} (2013) pp. 377--399.

\bibitem{lorz}
{\sc A. Lorz}, {\sc T. Lorenzi}, {\sc J. Clairambault}, {\sc A. Escargueil}, {\sc B. Perthame}.
\textit{Effects of space structure and combination therapies on phenotypic heterogeneity and drug resistance in solid tumors}, {\it in press}.

\bibitem{Lavi15122013}
{\sc O. Lavi}, {\sc J. M. Greene}, {\sc D. Levy}, {\sc M. M. Gottesman}. \textit{The role of cell density and intratumoral heterogeneity in multidrug
  resistance.} Cancer Research, \textbf{73(24)}(2013) pp. 7168--7175.

\bibitem{Greene2014627}
{\sc J. Greene}, {\sc O.~Lavi}, {\sc M.M. Gottesman}, {\sc D.~Levy}. \textit{The impact of cell density and mutations in a model of multidrug resistance in solid tumors}, Bulletin of Mathematical Biology, \textbf{76(3)}(2014) pp. 627--653.

\bibitem{2014arXiv1407.0865P}
{\sc G. G. Powathil}, {\sc M. A. Chaplain}, {\sc M. Swat}. \textit{Investigating the development of chemotherapeutic drug resistance in  cancer: A multiscale computational study}, IET Systems Biology, {\it in press}.



\end{thebibliography}
\end{document}